%% file: main.tex
\documentclass[11pt,letterpaper]{article}
\usepackage{fullpage}
\usepackage{amsmath, amssymb}
\usepackage{amsthm,nicefrac}
\usepackage{color}
\usepackage{hyperref}
\usepackage[compact]{titlesec}
\usepackage{verbatim, xspace}
\usepackage{thm-restate}
\usepackage{enumerate} 
\usepackage{subcaption}
\usepackage{tikz}
\usetikzlibrary{patterns}

\usepackage[bottom=1in,top=1in,left=1in,right=1in]{geometry}

\newtheorem{theorem}{Theorem} % [section]
\newtheorem{lemma}[theorem]{Lemma}
\newtheorem{corollary}[theorem]{Corollary}
\newtheorem{obsr}[theorem]{Observation}

\newtheorem{defn}[theorem]{Definition}
\newtheorem{claim}[theorem]{Claim} % [section]

\newcommand{\sm}{\setminus}
\newcommand{\eps}{\varepsilon}
\newcommand{\hh}{\mathcal{H}}
\newcommand{\ff}{\mathcal{F}}

\newcommand{\OPT}{\mathsf{OPT}}
\newcommand{\calF}{\mathcal{F}}

\newcommand{\VC}{\textsc{Vertex Cover}\xspace}
\newcommand{\FVS}{\textsc{Feedback Vertex Set}\xspace}

\newcommand{\HVD}{\textsc{$\hh$-Vertex Deletion}\xspace}
\newcommand{\HED}{\textsc{$\hh$-Edge Deletion}\xspace}
\newcommand{\PT}{\textsc{$k$-Path Transversal}\xspace}

\newcommand{\FVD}{\textsc{Planar-$\ff$ Vertex Deletion}\xspace}
\newcommand{\FED}{\textsc{Planar-$\ff$ Edge Deletion}\xspace}

\newcommand{\TVD}{\textsc{$k$-Treewidth Vertex Deletion}\xspace}
\newcommand{\PVD}{\textsc{$k$-Pathwidth Vertex Deletion}\xspace}
\newcommand{\TDVD}{\textsc{$k$-Treedepth Vertex Deletion}\xspace}
\newcommand{\TED}{\textsc{$t$-Treewidth Edge Deletion}\xspace}

\newcommand{\kVS}{\textsc{$k$-Vertex Separator}\xspace}
\newcommand{\kSVS}{\textsc{$k$-Subset Vertex Separator}\xspace}
\newcommand{\kSVSshort}{\textsc{$k$-SVS}\xspace}
\newcommand{\SVS}[1]{\textsc{$#1$-Subset Vertex Separator}\xspace}
\newcommand{\kES}{\textsc{$k$-Edge Separator}\xspace}
\newcommand{\kSES}{\textsc{$k$-Subset Edge Separator}\xspace}

\newcommand{\UML}{\textsc{Uniform Metric Labeling}\xspace}

\newcommand{\STCUT}{\textsc{Min }$s$-$t$\textsc{ Cut}\xspace}

\newcommand{\NPkSd}{\textsc{Noisy Planar $k$-SAT($\delta$)}\xspace}
\newcommand{\kSAT}{\textsc{$k$-SAT}\xspace}

\newcommand{\poly}{\mathrm{poly}}
\newcommand{\polylog}{\mathrm{polylog}}

\renewcommand{\epsilon}{\varepsilon}

\usepackage[colorinlistoftodos,textsize=tiny,textwidth=2cm,color=red!25!white,obeyFinal]{todonotes}
\newcommand{\agnote}[1]{\todo[color=blue!25!white]{AG: #1}\xspace}
\newcommand{\elnote}[1]{\todo[color=green!25!white]{EL: #1}\xspace}
\newcommand{\jlnote}[1]{\todo[color=red!25!white]{JL: #1}\xspace}

\newcommand{\initOneLiners}{%
    \setlength{\itemsep}{0pt}
    \setlength{\parsep }{0pt}
    \setlength{\topsep }{0pt}
}
\newenvironment{OneLiners}[1][\ensuremath{\bullet}]
    {\begin{list}
        {#1}
        {\initOneLiners}}
    {\end{list}}

\begin{document}

\title{{\bf Losing Treewidth by Separating Subsets}}

\author{ 
  Anupam Gupta\thanks{Computer Science Department, Carnegie Mellon University,
Pittsburgh, USA. Supported in part by NSF awards CCF-1536002, CCF-1540541,
and CCF-1617790, and the Indo-US Joint Center for Algorithms Under Uncertainty.}  \\ CMU 
  \and Euiwoong Lee\thanks{\tt euiwoong@cims.nyu.edu} \\ NYU 
  \and Jason Li$^*$ %\thanks{{\tt jmli@cs.cmu.edu} } 
  \\ CMU 
  \and Pasin Manurangsi\thanks{\tt pasin@berkeley.edu} \\ UC Berkeley 
  \and Micha{\l} W{\l}odarczyk\thanks{\tt m.wlodarczyk@mimuw.edu.pl} \\ University of Warsaw
}

\date{}

\maketitle

\thispagestyle{empty}
\begin{abstract}
\input{abstract}
\end{abstract}

\newpage

\setcounter{page}{1}

\section{Introduction}
\label{sec:intro}
\input{intro}
\input{techniques}
\input{related}

\section{Preliminaries}
\input{prelim}

\section{Vertex Deletion}
\label{sec:vertex}
\input{vertex}

\section{Edge Deletion}
\input{kses}

\subsection{\HED}
\input{hed}

\subsection{Inapproximability of \kES}
\input{inapprox}

{\small
\bibliographystyle{alpha}
\bibliography{refs}
}

\appendix
\input{twd-deletion}

\end{document}

%% file: abstract.tex
We study the problem of deleting the smallest set $S$ of vertices (resp.\ edges) from a given graph $G$ such that the induced subgraph (resp.\ subgraph) $G \sm S$ belongs to some class $\hh$. We consider the case where graphs in $\hh$ have treewidth bounded by $t$, 
and give a general framework to obtain
approximation algorithms for both vertex and edge-deletion settings from
approximation algorithms for certain natural graph partitioning problems
called \kSVS and \kSES, respectively. % \alert{Say something about $t$ and
  % $k$ here?}

For the vertex deletion setting, our framework combined with the current
best result for \kSVS, improves
approximation ratios for basic problems such as \TVD and
\FVD. %\alert{Explain these problems?} 
Our algorithms are simpler than previous works 
and give the first deterministic and uniform approximation algorithms under the natural 
parameterization.%\elnote{Slightly toned down.}

For the edge deletion setting, we give improved approximation algorithms
for \kSES combining ideas from LP relaxations and important separators. 
We present their applications in bounded-degree graphs, and 
also give an APX-hardness result for the edge deletion problems. 
% \alert{A bit  more here?}
%\elnote{Expanded a bit..}
%\mwnote{I wouldn't brag too much about the APX hardness here because
%  some hardness for PIT was known and the reduction to edge deletion is
%  easy} \agnote{Toned it down slightly.}

%% file: intro.tex
Let $\hh$ be a class of infinitely many graphs.  In the \HVD (resp. \HED) problem, we
are given a graph $G$ and we must find the smallest set of vertices (resp. edges) $X$
such that $G \sm X$ belongs to class $\hh$.  The simplest examples of
such problems are {\sc Vertex Cover} and {\sc Feedback Vertex Set}
problems, where $\hh$ is the set of all empty graphs (resp.\ forests)
and hence $G \sm X$ must exclude all edges (resp.\ cycles) of
$G$. 
Indeed, the problem has often been studied in the
context when $\hh$ is a class of graphs that {\em exclude} a fixed graph
$F$ in some sense (e.g., minor, subgraph, or induced subgraph).  In this case we let $\ff$ be a finite list of
excluded graphs and define $\hh$ to be the class of graphs that exclude
every graph from $\ff$. 
There has been a rich body of work studying {\em parameterized complexity} and {\em kernelization}
of these problems parameterized by the size of the optimal solution~\cite{HVJKV11, FLMS12, CM15, DDvH16, KLPRRSS16, EGK16, BBKM16, GLSS16, BCKP16b, JP17, ALMSZ17b, GJLS17}. 

In this paper, we focus on {\em parameterized approximation algorithms} parameterized by $\ff$ --- our desired running time is of the form $f(\ff) \cdot \poly(n)$.
Note that the notion of approximation is inherent in this
parameterization: even the simplest case of \textsc{Vertex Cover}, where the only graph in $\ff$ is a single edge and hence $f(\ff)$ is a constant, is NP-hard. 
Such an approximation algorithm could be used to obtain better kernels~\cite{ALMSZ17, ALMSZ17b}.
In addition to the $2$-approximation algorithms for {\sc Feedback Vertex
  Set}~\cite{BG96, BBF99, CGHW98}, the systematic study of the parameterized approximability depending on $\ff$ has also been done in the context of both parameterized algorithms and approximation algorithms~\cite{FJP10, FLMS12, GL15, BCKP16, JP17, KS17, BRU17, ALMSZ17, Lee18, KK18}. 

Many of the above algorithmic successes are based on one of the following two techniques.~\footnote{One notable exception that did not use any of the two techniques is the work of Fomin et al.~\cite{FLMS12}, which uses the notion of {\em lossless protrusion reduction} that reduces the instance size while preserving an approximation ratio.}

\begin{itemize}
  % \elnote{More intelligent remark on FLMS12?}  Recent work on
  % \textsc{Treewidth $t$-Deletion}~\cite{ALMSZ17, BRU17} focused on the
  % fact that, if $G$ has low treewidth, for any subset
  % $R \subseteq V(G)$ of vertices, there exists a small subset of
  % vertices $S$ such that $G \sm S$ does not have many vertices of $R$
  % in one component.
\item Specifically designed linear programming (LP) relaxations for the problem, often inspired by other classical optimization problems: 
These approaches were used in~\cite{FJP10, GL15, Lee18} to solve {\sc Diamond Hitting Set}, {\sc $k$-Star Transversal} and {\sc $k$-Path Transversal}, where the underlying classical problems were {\sc Feedback Vertex Set}, {\sc Dominating Set}, and {\sc Balanced Separator} respectively. While these tools often give principled ways to find optimal approximation ratios, the previous connections were tailored to specific settings.

%Establishing a connection to
%  classical optimization problems (e.g., {\sc Dominating Set} or graph
%  partitioning) and adapting tools from approximation algorithms
%  (e.g., linear programming (LP) relaxation or local search) to our
%  settings.  
  % \elnote{Difference to BLU, ALMSZ, FLRS here? Separate techniques
  % section?}
\item Combinatorial algorithms that exploit graph-theoretic structures of $\ff$-free graphs, using the general algorithm to find a balanced separator~\cite{FHL08} as a subroutine: 
This route was taken by~\cite{JP17, KS17, BRU17, ALMSZ17} for \FVD, {\sc Minimum Planarization}, {\sc Chordal Vertex Deletion}, and {\sc Distance Hereditary Vertex Deletion}.
While seamlessly bridging between graph-theoretic properties and
well-studied graph partitioning algorithms, this technique has the shortcoming that the best approximation factor for {\sc Balanced Separator} is $\Omega(\sqrt{\log n})$, forcing approximation ratios to depend on $n$.~\footnote{\cite{BRU17} did not use~\cite{FHL08} as a black-box and analyzed a linear programming relaxation directly, but only has an  approximation ratio of $\Omega(\log n)$. \cite{BRU17} and \cite{ALMSZ17} have an additional advantage of handling the weighted version, whereas our results cannot.}

%The use of graph-width parameters, most notably {\em
%    treewidth}. This is mainly used for the minor or immersion deletion
%  problems. For example, for \FVD, Fomin et al.~\cite{FLMS12}
%  crucially used the fact that any graph $G$ that does not have a planar graph
%  $F$ as a minor has bounded treewidth. This means that $G$ has a
%  rich combinatorial structure, which can be exploited by efficient
%  algorithms.
\end{itemize}

The natural question is: can we apply both kinds of techniques to give stronger results? 
We attempt to give a positive answer to this question, by defining a new kind of graph partitioning problems that (1) can be approximated well using some variants of the LP-based graph partitioning techniques and (2) exploit fundamental graph-theoretic concepts (e.g., treewidth) more closely than the traditional {\sc Balanced Separator} problem. 
We demonstrate the power of these approaches for minor or treewidth
deletion problems. We get simpler algorithms with better approximation ratios,
and we hope that more of such intimate connections can be made between structural graph theory and graph partitioning algorithms through intermediate problems.

\subsection{Our Results}
A class $\hh$ of graphs is called {\em hereditary} if $G \in \hh$, all its induced subgraphs are in $\hh$. 
Our main conceptual contribution is easy to state:
\begin{quote}
  \emph{If $\hh$ is hereditary, and if graphs in $\hh$ have bounded treewidth,
  good approximations for \HVD are implied by good
  approximations for a natural graph partitioning problem called \kSVS.}
\end{quote}

(A similar result holds for \HED and \kSES.) What are these partitioning
problems?  Given graph $G = (V(G), E(G))$, a subset
$R \subseteq V(G)$ of terminals, and an integer $k$, the \kSVS
(resp.\ \kSES) problem
asks to delete the minimum number of vertices (resp.\ edges) to
partition $G$ so that each component has at most $k$ vertices from
$R$. When $R = V(G)$, these problems are called 
\kVS (resp.\ \kES), and they generalize balanced separator
problems studied in the context of cuts and metrics~\cite{ENRS00}. The
case of general $R$ has close connections to problems such as
\textsc{Multiway Cut} and \UML.

Let us now elaborate on how we design approximation algorithms for these
separator problems, and how we develop the framework to use them to get new
results for \HVD and \HED.

%%%%%%%%%%%%%%%%%%%%%%%%%%%%
\iffalse
This problem has been addressed by Fomin et al.~\cite{FLMS12} in the case where $\hh$ is given as $\ff$-minor-free graphs with at least one planar graph in $\ff$.
This setting captures e.g., \textsc{Vertex Cover}, \textsc{Feedback Vertex Set}, and \textsc{Treewidth $t$-Deletion}.
Besides FPT algorithms and kernelization, they have shown a $c_{\ff}$-approximation algorithm for each such class.
The exact approximation ratio depends on $\ff$ and grows rapidly -- no upper bounds have been presented.

One interesting example of a class, for which a moderate approximation ratio is feasible, is given by graphs of bounded degree.
\HVD for $\hh$ being $d$-degree graphs admits $O(\log d)$-approximation~\cite{guruswami-inapproximability}.
This class is not minor-closed, however it is immersion-closed.
An example of a class without a constant ratio approximation is given by bipartite graphs --
under UGC there can be no $O(1)$-approximation for \textsc{Odd Cycle Transversal} [citation needed].
This class is hereditary but neither minor-closed nor immersion-closed.
\fi
%%%%%%%%%%%%%%%%%%%%%%%%%%%%

\subsubsection{Vertex Deletion Problems}
Our main result in this setting is a meta-theorem connecting \HVD
problems with the \kSVS (\kSVSshort) problem. 
In this paper, unless specified by a subscript, all constants hidden in $O(\cdot)$ are absolute constants. 
Also throughout the paper, $\OPT$ denotes the cardinality of the optimal solution of an optimization problem. 
An algorithm for \kSVSshort\ is called an {\em $(\alpha, \beta)$-bicriteria approximation algorithm}
if it returns a set $S \subseteq V(G)$ such that $|S| \leq \alpha \cdot \OPT$
and each component of $G \setminus S$ has at most $\beta k$ vertices
from $R$. This is a weaker requirement than a ``true'' approximation.
Also, when we refer to an exact algorithm for any optimization problem,
we mean an algorithm finding a solution of size $\OPT$.

\begin{restatable}{theorem}{mainvertex}
\label{thm:main_vertex}
Let $\hh$ be a hereditary class of graphs with treewidth bounded by $t - 1$.
Suppose that
\begin{OneLiners}
\item[(a)] \HVD admits an exact algorithm that runs in time $f(n, \OPT)$, and 
\item[(b)] \kSVS admits a $(\alpha(k), O(1))$-bicriteria approximation algorithm that runs in time $g(n, k)$  with $\alpha(k) = O(\log k)$. 
\end{OneLiners}
Then there is a $2 \alpha(t)$-approximation for \HVD with running time 
$\big( f(n, O(t \alpha(t))) + g(n, O(t \alpha (t))) \big) \cdot \log n$.
\end{restatable}

% We introduce our results on vertex deletion problems via \kSVS. 
% \agnote{Compress this para, and move to some later section?} 
%\emph{\color{gray}We use the following notion of a {\em bicriteria} approximation algorithm for \kSVS. Given an instance $G = (V(G), E(G))$ and $R \subseteq V(G)$ of \kSVS, an $(\alpha, \beta)$-bicriteria approximation algorithm for \kSVS returns a set $S \subseteq V(G)$ such that $|S| \leq \alpha \cdot \OPT$ and each component of $G \setminus S$ has at most $\beta k$ vertices from $R$. Note that even when $R = V(G)$ (called \kVS), the case $k = 1$ is exactly \VC, so it will not admit an $(2 - \eps, 1)$-bicriteria approximation algorithm under the Unique Games Conjecture~\cite{KR08}. Moreover, it was shown in~\cite{Lee18} that, when $k$ is allowed to grow with $n$, the problem becomes as hard to approximate as the \DkS problem; hence, it admits no $(n^{1/\poly \log \log n}, 1)$-bicriteria approximation algorithm in polynomial time unless the Exponential Time Hypothesis is false~\cite{Man17}.}

E.g., an $(O(1), O(1))$-bicriteria approximation for \kSVSshort 
would give an $O(1)$-approximation for \HVD running in $(g(n, O(t)) +
f(n, O(t))) \log n$ time. The best approximation for \kSVSshort
currently gives an $(O(\log k), 2)$-bicriteria approximation and
runs in time $n^{O(1)}$~\cite{Lee18}, so Theorem~\ref{thm:main_vertex} 
implies the following corollaries.  All algorithms in this paper are deterministic.~\footnote{The conference version of~\cite{Lee18} presents a randomized algorithm, but the journal version derandomized it.}

We first study {\em minor deletion}, where we want to exclude every graph in $\ff$ as a minor. 
%(See Section~\ref{subsec:related} for various other notions of exclusion.) 
The celebrated result by Fomin et al.~\cite{FLMS12} studied the case when $\ff$ has at least one planar graph (called \FVD), and gave a randomized $c_{\ff}$-approximation algorithm that runs in time $f(\ff) \cdot O(nm)$, where $c_{\ff}$ is a constant depending on $\ff$. While requiring the excluded family $\ff$ contain a planar graph seems restrictive, this case still captures fundamental optimization problems such as \textsc{Vertex Cover} and \textsc{Feedback Vertex Set}. 
The following corollary will be proved in Section~\ref{sec:vertex}. 

\begin{restatable}{corollary}{fvd}
\label{cor:fvd}
If $\hh$ is minor-closed and graphs in $\hh$ have treewidth at most $t$,
\HVD admits an $O(\log t)$-approximation algorithm that runs in time $O_\hh(n\log n) + n^{O(1)}$. 
In particular, \FVD admits an $O(\log f)$-approximation with running time $O_\ff(n\log n) +  n^{O(1)}$, where $f$ denotes the number of vertices of any planar graph in $\ff$.
\end{restatable}

The latter result improves the approximation ratio for \FVD from $c_{\ff}$ to $O(\log k)$. It is also deterministic, positively answering an open question in Kim et al.~\cite{KLPRRSS16}.
%Since this constant $c_{\ff}$
%depends on the Robertson-Seymour structure theorem, it is likely to be
%much larger than $O(\log k)$. 
% \agnote{Is this true? Also, can we say
%   something about the constant in the big-O for us?}
%   \mwnote{I would drop the sentence about RS. Our exact ratio is $16 \ln k$ but we could squeeze it to $(8+\eps) \ln k$}
%   \elnote{Removed the sentence. Constant $16$ is not too impressive to highlight?} 

Theorem~\ref{thm:main_vertex} also implies {\em uniform} algorithms with better running time and approximation ratio for 
many natural parameterized problems.  
A parameterized algorithm with parameter $k$ is called uniform if there is a single algorithm that takes an instance $I$ and a value of $k$ as input, and runs in time $f(k)\cdot |I|^{O(1)}$. 
Non-uniform algorithms for a parameterized problem indicate that there are different algorithms for each $k$, whose existence relies on non-constructive arguments.

If we have a sequence of families $\hh_k$ 
where every graph in $H_k$ has treewidth at most  $k$, 
and we know a uniform exact fixed-parameter tractable (FPT) algorithm for $\hh_k$-\textsc{Vertex Deletion}
parameterized by both $k$ and $\OPT$, we get a uniform
approximation algorithm for $\hh_k$-\textsc{Vertex Deletion}.
Some examples include \TVD,  \PVD, and \TDVD, where
$\hh_k$ is the set of all graphs with treewidth, pathwidth, and treedepth at most $k$ respectively. 
Another example is \PT where $\hh_k$ contains all graphs with
no simple path of length $k$ as a subgraph (or equivalently, as a
minor). For these cases we get the following results that will be proved in Section~\ref{sec:applications_vertex}.
%\agnote{Combine into one corollary?}
\begin{restatable}{corollary}{tvd}
\label{cor:tvd}
The following problems admit $O(\log k)$-approximations:
\begin{OneLiners}
\item[(a)] \TVD in time $2^{O(k^3\log^2 k)}\cdot n\log n + n^{O(1)}$,
\item[(b)]  \PVD in time $2^{O(k^2\log k)}\cdot n\log n + n^{O(1)}$,
\item[(c)]  \TDVD in time $2^{O(k^2\log k)}\cdot n\log n + n^{O(1)}$,
\item[(d)]  \PT in time $2^{O(k\log^2 k)}\cdot n\log n + n^{O(1)}$.
\end{OneLiners}
%\TVD, \PVD, and \TDVD admit an $O(\log k)$-approximation with running times re $2^{O(k^2\log k)}\cdot n\log n + n^{O(1)}$.
%\end{restatable}
%\begin{restatable}{corollary}{pt}
%\label{cor:pt}
%\PT admits an $O(\log k)$-approximation with running time $2^{O(k\log^2 k)}\cdot n\log n + n^{O(1)}$.
\end{restatable}

The first result of Corollary~\ref{cor:tvd} improves the best current approximation for \TVD
from~\cite{FLMS12}, which did not explicitly state the dependency of the approximation ratio and
running time on the treewidth $k$. 
Our algorithms also give the first uniform $O_k(1)$-approximation algorithms in this parameterization, 
while the protrusion-replacement technique of~\cite{FLMS12} makes their
algorithms non-uniform.
For \TDVD a $2^k$-approximation algorithm has been known~\cite{treedepth-apx}. The last result
of Corollary~\ref{cor:tvd} tightens the runtime of the $O(\log k)$-approximation algorithm for \PT
that runs in time $O(2^{k^3 \log k}n^{O(1)})$~\cite{Lee18}.

% We first demonstrate this technique for \TVD, which is given by $\hh_k$ being the family of all graphs with their treewidth bounded by $k$.
% The best known approximation algorithm for \TVD follows from the result for \FVD by~\cite{FLMS12}.
% The dependency of the approximation ratio and running time on the parameter $k$ has not been analyzed.
 
% \begin{restatable}{corollary}{tvd}
% \label{cor:tvd}
% The \TVD problem admits an $O(\log k)$-approximation with running time $2^{O(k^3\log^3 k)}\cdot n\log n + n^{O(1)}$ .
% \end{restatable}

% The \PT problem is another special case of $\hh_k$-\textsc{Vertex Deletion} where $\hh_k$ is the family of all graphs that does not have the simple path with $k$ vertices as a subgraph (or equivalently as a minor). 
% Theorem~\ref{thm:main_vertex} also improves the running time of the $O(\log k)$-approximation algorithm for \PT  that runs in time $O(2^{k^3 \log k}n^{O(1)})$~\cite{Lee18} in terms of $k$.
% %(The exponent of $n$ stays the same.)
% %\mwnote{The former exponent came from solving LP and now we have $f(k)n\log n + LP$ what seems to be an improvement.} 
% \begin{restatable}{corollary}{pt}
% \label{cor:pt}
% The \PT problem admits an $O(\log k)$-approximation with running time $2^{O(k\log^2 k)}\cdot n\log n + n^{O(1)}$ .
% \end{restatable}

Another application of Theorem~\ref{thm:main_vertex} arises in bidimensionality theory.  
%We know that 
For any family $\hh$ with treewidth bounded by $k$, an 
efficient polynomial-time approximation scheme (EPTAS) for \HVD on excluded-minor families was first given by Fomin et al.~\cite{FLRS11}.
For \HVD on $M$-minor free graphs $G$, the algorithm runs in time $n^{f(\hh, M)}$~\cite{FLRS11}, or becomes non-uniform and runs in time $2^{f(\hh, M)} \cdot n^{O(1)}$~\cite{FLMS12}. 
%, where the the dependency on $k$ in the runtime was not stated explicitly.
% and relied on the MSO logic theory. %remained elusive 
%and the algorithms were non-uniform. 
%Moreover, the EPTAS for, e.g., \TVD was non-uniform.
In Section~\ref{sec:bidimensionality} we prove the following meta-theorem that gives faster algorithms for many relevant problems, which are uniform as long as the promised exact algorithm is.
% directly applies to many relevant problems.
%\elnote{Slightly changed the above two paras.}
%For \HVD on $M$-minor free graphs $G$, the algorithm runs in time $n^{f(\hh, M)}$~\cite{FLRS11}, or becomes non-uniform and runs in time $2^{f(\hh, M)} \cdot n^{O(1)}$~\cite{FLMS12}. 

\begin{restatable}{theorem}{eptas}
\label{thm:eptas}
Let $\hh$ be a hereditary class of graphs with treewidth bounded by
$t - 1$, and $M$ be a fixed graph.  Suppose that \HVD admits an exact
algorithm that runs in time $f(n, \OPT)$.  Then \HVD admits an $(1+\eps)$-approximation algorithm on
$M$-minor-free graphs with running time
$f(n, O_M(\nicefrac{(t\log^3 t)}{\epsilon^2}))\cdot n\log n + n^{O(1)}$.
\end{restatable}

\subsubsection{Edge Deletion Problems}

Our main result for the edge deletion problems is the following. 
Unlike the vertex version, the approximation ratio becomes an absolute constant, 
but the algorithm uses the maximum degree of $G$, $\deg(G)$, as an additional parameter.

\begin{restatable}{theorem}{mainedge}
\label{thm:main_edge}
Let $\hh$ be a class of graphs closed under taking a subgraph. 
Suppose graphs in $\hh$ have treewidth bounded by $t - 1$, and \HED admits an exact algorithm with running time
$f(n, \OPT)$.
Then there is an $(3 + \eps)$-approximation for \HED with running time \mbox{$\big(\min\{ 2^{O(t^2 \deg(G)^3 /\eps^3)}n^{O(1)}, n^{O(t \deg(G) / \eps )} \} + f(n, t \deg(G) / \eps) \big) \log(n / \eps)$}.
\end{restatable}
%\elnote{Should we set $\eps = 1$ and remove them?}

The above theorem is based on our improved results on \kSES. 
The previous best approximation algorithm for \kSES was an $O(\log k)$-approximation that runs in time $n^{O(1)}$~\cite{Lee18}. 
While the existence of an $O(1)$-approximation algorithm that runs in time $n^{O(1)}$ would refute the Small Set Expansion Hypothesis~\cite{RST12}, we show that one can get significantly better approximations factor using $k$ as a parameter. 
%We introduce our results on edge deletion problems via \kSES. 
%Recall that given a graph $G = (V(G), E(G))$ and $R \subseteq V(G)$, \kSES asks to delete the minimum number of edges such that each connected component has at most $k$ vertices from $R$.
%The special case when $R = V(G)$ is called \kES. 
%
\begin{restatable}{theorem}{kes}
\label{thm:kes}
The following parameterized %approximation
algorithms
for \kSES exist:
%\mwnote{changed first (i): we have to mention that the
%  first result holds for \kES, but we have not introduced this problem
%  yet}\agnote{Oops, thanks.}
\begin{OneLiners}
\item[(a)] a $(2 + \eps)$-approximation that runs in time $2^{O(k \log (k/\eps))} n^{O(1)}$ for the case $R = V(G)$,
\item[(b)] a $2$-approximation that runs in time $n^{k+O(1)}$, and
\item[(c)] a $(2+\eps)$-approximation that runs in time $2^{O(k^2 \deg(G) /\eps)} n^{O(1)}$.
\end{OneLiners}
\end{restatable}

\iffalse
We show the following algorithms that achieve (almost) 2-approximation. 
\begin{restatable}{theorem}{kesgeneral}
\label{thm:kses_general}
There is a $2$-approximation algorithm for \kSES that runs in time $n^{k+O(1)}$. \jlnote{I changed this to $n^{k+O(1)}$ since it looks slightly better}
\end{restatable}

\begin{restatable}{theorem}{kesbounded}
\label{thm:kses_bounded}
There is a $(2+\eps)$-approximation algorithm for \kSES that runs in time $2^{O(k^2 \deg(G) /\eps)} n^{O(1)}$.
\end{restatable}
\elnote{Should I combine them as one thm?} \agnote{I agree}
\fi

We now present corollaries of Theorem~\ref{thm:main_edge}.
Applying exact algorithms (parameterized by $\OPT$) for well-known cases immediately imply the following corollary proved in Section~\ref{subsec:edge_applications}. 

\begin{restatable}{corollary}{fed}
\label{cor:fed}
If $\hh$ is minor-closed and with treewidth bounded by $t$,
\HED admits a $(3 + \eps)$-approximation algorithm that runs in time $f(t, \deg(G), \eps) \cdot n^{O(1)}$ for some function $f$. 
In particular, \FED admits a $(3 + \eps)$-approximation algorithm with running time $f(\ff, \deg(G), \eps) \cdot n^{O(1)}$.
\end{restatable}

In Section~\ref{subsec:edge_applications}, We also present implications of Theorem~\ref{thm:main_edge} to the \NPkSd problem studied by Bansal et al.~\cite{BRU17}. For a fixed $k = O(1)$, an instance of \NPkSd is an instance of $\phi$ of \kSAT with $n$ variables and $m$ clauses where the {\em factor graph} of $\phi$ becomes planar after deleting $\delta m$ edges (See Section~\ref{subsec:edge_applications} for formal definitions).
Bansal et al.~\cite{BRU17} proved that for any $\eps > 0$, there is an algorithm that achieves $(1 + O(\eps + \delta \log m \log \log m))$-approximation in time $m^{O(\log \log m)^2 / \eps}$. We prove that if the degree of the factor graph is bounded, we can obtain an improved algorithm.
Note that \kSAT with the maximum degree $O(1)$ has been actively studied and proved to be APX-hard (e.g., \textsc{$3$-SAT(5)}) for general factor graphs. 

\begin{restatable}{corollary}{bru}
\label{cor:bru}
For any $\eps > 0$, there is an $(1 + O(\eps + \delta))$-approximtion algorithm for \NPkSd that runs in time $f(\eps, \deg(\phi)) \cdot m^{O(1)}$ for some function $f$, where $\deg(\phi)$ indicates the maximum degree of the factor graph of $\phi$.  
\end{restatable}

For the edge deletion problems, the trivial reduction from \VC, that gave $(2 - \eps)$-inapproximability for all vertex problems under the Unique Games Conjecture, does not work. 
When $k = 1$, \kSES becomes the famous {\sc Multiway Cut} problem, which is hard to approximate within a factor $\approx 1.2$ assuming the Unique Games Conjecture~\cite{AMM17}. 
One can also speculate that the edge deletion problems may have an exact algorithm or PTAS when $\deg(G)$ is bounded. 
We prove the following hardness result that \kES with $k = 3$ is APX-hard even when $\deg(G) = 4$. 
Taking $\ff$ to be the set of all graphs with three vertices (which are all planar), this also proves that excluding $\ff$ as a subgraph, minor, or immersion will not admit a PTAS even for bounded degree graphs. 

\begin{restatable}{theorem}{hardnessedge}
\label{thm:kes_inapprox}
There exists a constant $c > 1$ such that \kES is NP-hard to approximate within a factor of $c$ even when $k = 3$ and $\deg(G) = 4$.
Consequently, when $\ff$ is the set of all graphs with three vertices, deleting the minimum number of edges to exclude $\ff$ as a subgraph, minor, or immersion is APX-hard for bounded degree graphs. 
\end{restatable}
We note that the above hardness result only leaves open the case of $\deg(G) = 3$; when $\deg(G) = 2$, the graph is simply a disjoint union of paths and cycles, and hence \kES can be solved (exactly) in polynomial time.
%\pmnote{Should we mention that $k$-SES generalizes Multiway Cut and hence is UG-hard to approx to 1.2? It's a bit strange to leave this out, since the positive results only use $k$-SES...}
%\elnote{Have some deleted hardness texts. Need them here? }

%%% Local Variables:
%%% mode: latex
%%% TeX-master: "main"
%%% End:

%% file: techniques.tex
\subsection{Techniques}
\label{sec:techniques}
\paragraph{Vertex Deletion.} 
We briefly sketch our proof techniques for Theorem~\ref{thm:main_vertex} for general \HVD using an algorithm for \kSVS as a black box. For simplicity, let us focus on \TVD. 

Let $S^* \subseteq V$ be the optimal solution with $|S^*| = \OPT$. 
Our high-level approach is the following iterative algorithm that maintains a feasible solution $R \subseteq V$ and {\em refines} it to a smaller solution in each iteration. (Initially we start from $R = V$.) 
The simple but crucial lemma for us is Lemma~\ref{lem:cut} in Section~\ref{sec:vertex}, which states that if in the induced subgraph $G[V \setminus S^*]$, which has treewidth at most $k$, 
there are at most $\eps |R|$ vertices such that additionally deleting them from $G[V \setminus S^*]$ ensures that each connected component has at most $O(k / \eps)$ vertices from $R$.
This type of argument guaranteeing the existence of a small separator that {\em finely separates} a subset (i.e., each component has $O_{k, \eps}(1)$ vertices from $R$) previously appeared in Fomin et al.~\cite{FLRS11} for bidimensionality theory. 
Our lemma admits a simpler proof because we need less properties and do not need to be constructive.% (We do not know $S^*$ in the first place.)
%\mwnote{This is not always true -- in Theorem~\ref{thm:eptas} we need the constructive version}

Our main conceptual contribution that bridges treewidth deletion and \kSVS is to observe that the above lemma guarantees a feasible solution of \SVS{O(k/\eps)} of size $\OPT + \eps R$. 
Applying an $(\alpha, \beta)$-bicriteria approximation algorithm for \kSVS will delete at most $\alpha(\OPT + \eps R)$ vertices to make sure that each connected component has at most $O(\beta k / \eps)$ vertices from $R$. Since $R$ is a feasible solution, each connected component admits a solution of size $O(\beta k / \eps)$, which can optimally solved in $f(\beta k / \eps) \cdot n^{O(1)}$ time by deleting at most $\OPT$ vertices. In total, the size of our new solution is at most $\alpha(\OPT + \eps R) + \OPT$. 
By appropriately adjusting $\eps$, we can prove that unless $R = O(\alpha \cdot \OPT)$, the size of the new solution is at most $|R| / 2$, which implies that we will achieve $O(\alpha)$-approximation in at most $O(\log n)$ iterations. The current best $(O(\log k), 2)$-bicriteria approximation algorithm for \kSVS immediately yields $O(\log k)$-approximation for \TVD. 

Recently, Bansal et al.~\cite{BRU17} and Agarwal et al.~\cite{ALMSZ17}
used graph partitioning algorithms to solve treewidth deletion problems. 
Agarwal et al.'s approach was based on graphs with bounded treewidth admitting small {\em global separators} (i.e., whose deletion ensures each component has $2n/3$ vertices),
while Bansal et al.\ additionally used the fact that any subset $R \subseteq V$ admits a small {\em $R$-global separator} (i.e., whose deletion ensures each component has $2|R|/3$ vertices). 
Such small separators are found by a modification of traditional graph
partitioning algorithms for global separators, which allows us to
recurse into smaller components.
Using such global partitioning algorithms
gives an inherent loss  of $\Omega(\log n)$.
Our results indicate that computing a {\em finer-grained separator} of
$R$, such as \kSVS, avoids the loss in terms of $n$. One downside of our
approach is that it does not work for weighted settings.
The idea of reducing to the {\em subset} version of a classical combinatorial optimization problem was also employed by Bonnet et al.~\cite{BBKM16} where they used {\sc Subset Feedback Vertex Set} to solve {\sc Bounded $\mathcal{P}$-Block Vertex Deletion}.

It is an interesting open question to see whether \kSVS admits an $(\alpha, \beta)$-bicriteria approximation algorithm for absolute constants $\alpha, \beta$ since it will imply an $O(\alpha)$-approximation algorithm for \TVD that does not depend on $k$ by Theorem~\ref{thm:main_vertex}. 
The best inapproximability is $(2 - \eps)$ coming from \VC, and this is even for the case $R = V$.

\paragraph{\kSES.}

Here we highlight our techniques for the edge deletion problems, which
result in algorithms with better approximation factors than their vertex
deletion counterparts. The gap in difficulty between vertex- and edge-deletion versions has been observed in other cut problems on undirected graphs, such as \textsc{Multiway Cut}~\cite{garg2004multiway,buchbinder2017simplex} and \textsc{Minimum $k$-way Cut}~\cite{saran1995finding}.\footnote{The hardness of node $k$-way cut follows from the observation that the instance is feasible iff there is an independent set of size $k$ in the graph, and independent set is hard to approximate~\cite{zuckerman2006linear}.} Intuitively, the reason is that in edge deletion problems, we can charge the solution cost to the boundary size of each connected component in the remaining graph (without the deleted edges). Since every edge deleted belongs to the boundary of exactly two components, the sum of the boundary sizes of the components is exactly twice the solution cost. Charging the cost of an algorithm to the sizes of boundaries proves to be a more tractable strategy in many cases.

The \kES problem is a special case of \kSES where $R = V(G)$.
%\mwnote{added \kES definition}
Our two $(2+\eps)$-approximation algorithms for \kES start by reducing the degree of the graph to $O(k)$, while sacrificing only an $(1+\eps)$ factor loss in approximation. This step relies on the observation that if a vertex has very high degree, then nearly all of its incident edges must be deleted in any feasible solution, so we might as well delete them all and sacrifice an $(1+\eps)$ factor loss.

After the degree of the graph is parameterized by $k$, our first algorithm begins with any feasible solution and iteratively improves it using local search. At each step, the algorithm examines connected components of at most $k$ vertices and looks for one which can improve the current solution. If the graph has degree $O(k)$, then there are only $k^{O(k)}n$ many connected components of size at most $k$, which leads to a running time FPT in $k$.

The second algorithm for \kES relies on a direct reduction to an instance of \UML by viewing each of the $k^{O(k)}n$ connected components of size at most $k$ as a color in \UML. It then applies the 2-approximation algorithm of \UML from~\cite{KT02}.

For the more general \kSES problem, the local search algorithm generalizes to one running in time $n^{k+O(1)}$. Instead of trying all connected components of size at most $k$, we try all subsets of $R$ of size at most $k$, not necessarily connected. Determining if a given subset improves the solution is more technical, requiring a gadget reduction to a minimum $s$-$t$ cut instance.

Finally, for the case when the graph degree is small, we can follow the \UML reduction approach to obtain an algorithm FPT in both $k$ and $\deg(G)$. In particular, we prove that, modulo an $(1+\eps)$ loss in approximation, there are essentially $2^{O(k^2\deg(G))}n$ many relevant connected components to consider. For this, we use the idea of \textit{important cuts}, a tool that has been used in FPT algorithms for other cut problems, such as \textsc{Multiway Cut}~\cite{marx2006parameterized,chen2009improved}. Assigning each of these relevant components a color gives a \UML instance of size parameterized by both $k$ and $\deg(G)$, which is again approximated to factor $2$.

%%% Local Variables:
%%% mode: latex
%%% TeX-master: "main"
%%% End:

%% file: related.tex
\subsection{Related Work}
\label{subsec:related}
In this subsection, we briefly survey known parameterized approximation algorithms for \FVD and \FED, parameterized by $\ff$,
leaving out a rich set of results on exact parameterized algorithms and kernelization (often parameterized by $OPT$)~\cite{HVJKV11, FLMS12, CM15, DDvH16, KLPRRSS16, EGK16, BBKM16, GLSS16, BCKP16b, JP17, ALMSZ17b, GJLS17}. 

For minor deletion problems, 
Agrawal et al.~\cite{ALMSZ17} gave a $\polylog(n)$-approximation for \FVD.
When $\ff = \{ K_5, K_{3, 3} \}$, this problem is known as \textsc{Minimum Planarization} and admits $\polylog(n)$-approximation in time $n^{O(\log n / \log \log n)}$~\cite{KS17}. 
\cite{FJP10} gave a $9$-approximation for \textsc{Diamond Hitting Set}, which is excluding the graph with two vertices and three parallel edges as a topological minor. 
%while many of the aforementioned work~\cite{FLMS12, KS17, GPRTW17, Lee18} and our results hold only for the unweighted case, there is a large volume of work on the weighted case. For \TED and \FVD, Agarwal et al.~\cite{ALMSZ17} gave an $O(\log^{1.5} n)$-approximation algorithm, which was improved to $O(\log n \log \log n)$-approximation by Bansal et al.~\cite{BRU17}. 
Besides \VC ($F$ is a single edge) and \FVS ($F$ is a triangle), to the best of our knowledge, 
the only special cases of $\FVD$ that admit an $O_{\calF}(1)$-approximation algorithm for the weighted case is when $F$ is a diamond~\cite{FJP10} or 
$F$ is a simple path (where minor deletion and subgraph deletion become equivalent). 
%\pmnote{Should we mention that all problems considered are W[1]-hard, and hence we resort to FPT approx?}

When we exclude a single graph $F$ as a
  subgraph, there is a simple $k$-approximation algorithm where $k$ is
  the number of vertices in $F$. A nearly-matching hardness was proved
  by 
Guruswami and Lee~\cite{GL15}, who showed that \HVD is NP-hard to
approximate within a factor of $k - 1 - \eps$ for any $\eps > 0$ ($k -
\eps$ assuming the Unique Games Conjecture) whenever $F$ is $2$-vertex-connected. If $F$ is a star or a simple path with $k$ vertices, $O(\log k)$-approximation algorithms are known~\cite{GL15, Lee18}. %The algorithm for the $k$-path runs in time $2^{O(k^3 \log k)}\cdot n^{O(1)}$. 

For \HED, the notion of {\em immersion deletion} has commonly been studied instead of minor deletion~\cite{GPRTW17}. Bansal et al.~\cite{BRU17} gave an $O(\log n \log \log n)$-approximation for \TED.
The edge-deletion version for {\em induced subgraph deletion} was also studied~\cite{BCKP16}. 

There is also vast literature on general \HVD or \HED besides the aforementioned minor, immersion, subgraph, and induced subgraph deletions.
Lund and Yannakakis~\cite{LY93} considered the {\em maximization version} where we want to find the maximum $S \subseteq V(G)$ such that the induced subgraph $G[S] \in \hh$, and showed that whenever $\hh$ is hereditary and {\em nontrivial} ($\hh$ contains an infinite number of graphs and does not contain an infinite number of graphs), then the maximization version is hard to approximate within a factor $2^{\log^{1/2 - \eps} n}$ for any $\eps > 0$. This inapproximability ratio was subsequently improved to $n^{1 - \eps}$ for any $\eps > 0$ by Feige and Kogan~\cite{FK05}.

{\sc Chordal Vertex Deletion}~\cite{ALMSZ17, JP17, KK18} and {\sc Odd Cycle Transversal}~\cite{ACMM05} are other primary examples of $\HVD$;  they can be captured as a subgraph deletion when $\ff$ is the set of all chordless or odd cycles. The problem of reducing other width parameters (e.g., rankwidth, cliquewidth) have been studied~\cite{ALMSZ17}. Besides approximation algorithms, these problems also have been studied through the lens of their parameterized complexity (parameterized by $\OPT$) and covering-packing duality (known as the Erd\H{o}s-P\'{o}sa property). We refer the reader to the introduction of~\cite{ALMSZ17} and~\cite{GL15} for more detailed survey. 

\iffalse
Agrawal et al.~\cite{ALMSZ17} gave $\polylog(n)$-approximations {\sc Chordal Vertex Deletion}, and {\sc Distance Hereditary Vertex Deletion}. 
Jansen and Pilipczuk~\cite{JP17} and Kim and Kwon~\cite{KK18} obtained $\poly(OPT)$-approximations for {\sc Chordal Vertex Deletion}, by proving Erd\H{o}s-P\'{o}sa type properties on packing and covering holes. 

Another immediate open problem is to study the approximability of \kSVS, whose better approximation ratio will automatically improve approximation ratios for other problems via Theorem~\ref{thm:main}. 
Currently the best known hardness factor is from \VC, which is $(2 - \eps)$ even assuming the Unique Games Conjecture. 
\fi

%%% Local Variables:
%%% mode: latex
%%% TeX-master: "main"
%%% End:

%% file: prelim.tex
Unless otherwise specified by a subscript, all constants hidden in $O(\cdot)$ notations are absolute constants that do not depend on any parameter. 
For a graph $G = (V(G), E(G))$, let $n$ denote the number of vertices, and a subset $S$ of vertices or edges, let $G \setminus S$ be the graph after deleting $S$ from $G$. 
For disjoint subsets $C_1, \dots, C_m \subseteq V(G)$, let $E(C_1, \dots, C_m) := \{ (u, v) \in E(G): u \in C_i, v \in C_j, i \neq j \}$. 
For $C \subseteq V(G)$, let $\partial(C) := E(C, V \setminus C)$. 
For $v \in V$, let $\deg(v)$ denote the degree of $v$, and let $\deg(G)$ be the maximum degree of $G$.

\paragraph{Treewidth and pathwidth.}
Given a graph $G = (V(G), E(G))$, a tree $T = (V(T), E(T))$ is called a {\em tree decomposition} of $G$ if every node (also called a bag) $t \in V(T)$ is a subset of $V(G)$, and the following conditions are met.
\begin{OneLiners}
\item[1.] The union of all bags is $V(G)$. 
\item[2.] For each $v \in V(G)$, the subtree of $T$ induced by $\{ t \in V(T) : v \in t \}$ is connected. 
\item[3.] For each $(u, v) \in E(G)$, there is a bag $t$ such that $u, v \in t$. 
\end{OneLiners}
The {\em width} of $T$ is the cardinality of the largest bag minus $1$, and the {\em treewidth} of $G$, denoted $tw(G)$, is the minimal width of a tree decomposition of $G$.
If we restrict the tree $T$ to be  a path, we obtain analogous notions of \emph{path decomposition} and \emph{pathwidth} of $G$, denoted $pw(G)$.

\paragraph{Treedepth.}
A \emph{treedepth decomposition} of $G$ is a tree $T$ with an injective mapping $\phi: V(G) \rightarrow V(T)$, such that whenever $(u,v) \in E(G)$ then $\phi(u)$ and $\phi(v)$ are in ancestor-descendant relation.
The \emph{treedepth} of $G$, denoted $td(G)$, is the minimum
height of a treedepth decomposition of G.
We have $tw(G) \le pw(G) \le td(G) - 1$~\cite{treedepth-dynamic}.

\paragraph{Minors.}
We say that graph $M$ is a \emph{minor} of graph $G$ if there exists a mapping $\phi$ from $V(M)$ to disjoint connected subgraphs of $G$, such that whenever $(u,v) \in E(M)$ then $E(\phi(u),\phi(v)) \ne \emptyset$.
Otherwise we say that $G$ is $M$-minor-free.
If $M$ is planar, then all $M$-minor-free graphs have treewidth bounded by $|V(M)|^{O(1)}$~\cite{RS86, CC16}.

%\pmnote{Definition of Immersion \& Hereditary?}
%\elnote{More?}
%\elnote{Organization? Currently in each section.}

%% file: vertex.tex
In this section, we prove our results for the vertex deletion problems. 
We first prove Theorem~\ref{thm:main_vertex} and then show
its applications to \FVD, uniform algorithms, and bidimensionality.

Our proof of Theorem~\ref{thm:main_vertex} is based on the following simple lemma that reveals a natural connection between \kSVS and \HVD when graphs in $\hh$ have bounded treewidth. 
%\elnote{Mention FLRS (EPTAS) paper?}
\begin{lemma}\label{lem:cut}
Suppose graph $G$ has its treewidth bounded by $t-1$ and let $R \subseteq V(G)$.
Then for each natural number $\delta$ there exists a set $X \subseteq V(G)$ such that
$|X| \le \frac{t}{\delta} \cdot |R|$ and each connected component of $G \sm X$
contains at most $\delta$ elements from $R$.
What is more, if the tree decomposition is given, the set $X$ can be constructed in polynomial time.
\end{lemma}

\begin{proof}
Consider a tree decomposition of $G$ of width $t-1$.
For a bag $B$ let $r(B)$ denote the number of vertices from $R$ introduced in the subtree of the
decomposition rooted at $B$.
If $|R| \le \delta$ then $X = \emptyset$ satisfies the claim and otherwise there is a bag $B$ with
$r(B) > \delta$.
Let $B_0$ be such a bag with all its descendant having $r(B) \le \delta$.
Vertices contained in $B_0$ form a cut with all connected components formed by descendants of $B_0$ 
having at most $\delta$ vertices from $R$.
We iterate this procedure and define $X$ to be the union of all performed cuts.
Each cut is formed by at most $t$ vertices and there can be at most $\frac{R}{\delta}$ iterations.
The claim follows.
\end{proof}

Now we recall and prove the main theorem for the vertex deletion problems.

\mainvertex*

\begin{proof}
Let $\eps > 0$ be a constant determined later (depending on $t$). 
Our algorithm maintains a feasible solution $R \subseteq V(G)$ (say we start from $R = V(G)$) and iteratively finds a smaller solution. 
Let $S^* \subseteq V(G)$ be an optimal solution to \HVD and let $R \subseteq V(G)$ be the current solution. 
The graph $G \sm S^*$ has its treewidth bounded by $t - 1$, therefore
Lemma~\ref{lem:cut} with $\delta = t / \eps$ guarantees that there exists 
a set $X \subseteq V(G) \sm S^*,\, |X| \le \eps |R|$ so that each connected component in $G \sm (S^* \cup X)$
has at most $t / \eps$ vertices from $R$.

We launch the $(\alpha(k), O(1))$-bicriteria approximation for \kSVS on $G$ with $k = t / \eps$ and $\alpha(k) = O(\log k)$.
It returns a set $Y \subseteq V(G)$ of size at most
$$\alpha \cdot |S^* \cup X| \le \alpha \cdot (\OPT + \eps |R|)$$
such that each connected component of $G \sm Y$ has at most $O(t / \eps)$ vertices from $R$.
Since $\hh$ is hereditary, $R \cup Y$ is a valid solution and we have a bound $O(t / \eps)$ on the solution size for each connected component.
We thus can solve \HVD on each component $C \subseteq G \sm Y$ in time $f(n, O(t / \eps))$.
We know that $C \cap S^*$ is a feasible solution for each $C$, so the sum of solution sizes is bounded by $|S^*| = \OPT$.

Let $R'$ be the union of $Y$ and all solutions obtained for the connected components in $G \sm Y$.
It will be the new $R$ in the next iteration. 
%Assuming $R$ is an $\alpha$-approximate solution,
Since
$|R'| \leq |Y| + \OPT \leq (\alpha + 1) \OPT + \alpha\eps |R|,$
as long as 
$$
(\alpha(k) + 1) \cdot \OPT + \alpha(k) \eps |R| \leq (3/4) |R| 
\quad \Leftrightarrow \quad |R| \geq \frac{\alpha(k) + 1}{3/4 - \alpha(k) \eps} \OPT, 
$$
the size of the maintained solution is decreased by a factor of $3/4$. 
Since $\alpha(k) = O(\log k)$, if $\eps = c / \alpha(t)$ for small constant $c > 0$, 
$$
\frac{\alpha(k) + 1}{3/4 - \alpha(k) \eps} = 
\frac{\alpha(t / \eps) + 1}{3/4 - \alpha(t / \eps) \eps} =
\frac{\alpha(t \alpha(t) / c) + 1}{3/4 - c \alpha(t \alpha(t) / c) / \alpha(t)} 
\leq 2 \alpha(t).
$$
The last inequality holds since 
$\alpha(t \alpha(t) / c) \leq \alpha(t) + \alpha(\alpha(t) / c)$ gets multiplicatively closer to $\alpha(t)$ as $t$ grows,
so that for small enought $c > 0$, we can ensure that the denominator is at least $3/5$, and the numerator is at most $(6/5)\alpha(t)$ for large enough $t$. 
Therefore, if we begin with $R = V(G)$ and iterate the procedure $O(\log n)$ times, we have a $2\alpha(t)$-approximation. 
The running time for each iteration is $f(n, k) + g(n, k) = f(n, O(t \alpha (t))) + g(n, O(t \alpha(t)))$. 
\end{proof}

We combine this meta-theorem with a recent  result for \kSVS.

\begin{theorem}[\cite{Lee18}]
There exists an $(O(\log k), \,2)$-bicriteria approximation algorithm for \kSVS that runs in time $n^{O(1)}$. 
\label{thm:ksvs}
\end{theorem}

\begin{corollary}
\label{cor:main_vertex}
Suppose $\hh$ is a hereditary class of graphs with treewidth bounded by $t - 1$
and \HVD admits an exact algorithm that runs in time $f(n, \OPT)$.
Then \HVD admits $O(\log t)$-approximation algorithm with running time $f(n, O(t\log t))\log n + n^{O(1)}$.
\end{corollary}
%\mwnote{This is the result we need in applications and I suggest making this one main theorem and (\ref{thm:main_vertex}) a lemma. We might also move Theorem~\ref{thm:ksvs} here for legibility.}

We are ready to present the most general result
improving upon  \cite{FLMS12} who gave a $c_{\hh}$-approximation for \HVD for some implicit constant $c_{\hh}$.
We emphasize that the constant hidden in term $O(\log t)$ is universal.

\fvd*
\begin{proof}
For \FVD, we use the Polynomial Grid Minor theorem~\cite{CC16} which says that if $G$ does not have a planar graph $F$ as a minor, the treewidth of $G$ is bounded by $|V(F)|^{O(1)}$.
\FVD admits a linear-time exact algorithm parameterized by the solution size~\cite{bodlaender97}  so the assumptions of Corollary~\ref{cor:main_vertex} are satisfied.

Due to the result of Robertson and Seymour~\cite{graph-minors-xx}, every minor-closed class can be represented
as $\ff$-minor-free graphs for some finite family $\ff$.
If the treewidth in $\hh$ is additionally bounded, then
at least one of graphs in $\ff$ must be planar.
Therefore \HVD reduces to \FVD.
\end{proof}

\subsection{Uniform Algorithms for Width Reduction}
\label{sec:applications_vertex}
%\elnote{Is uniform most representative description? }
Another application of our approach emerges
when we deal with a sequence of families $\hh_k$.
In~contrary to the previously known techniques, Theorem~\ref{thm:main_vertex} can produce uniform algorithms for $\hh_k$-\textsc{Vertex Deletion}
when provided with an exact uniform algorithm parameterized by both $k$ and $\OPT$.
%In \TVD, \PVD, and \TDVD, the family $\hh_k$ stands for graphs with respectively treewidth, pathwidth, or treedepth bounded by $k$.
We present such an exact algorithm for \TVD, together with related problems,
and combine it with our framework.
Then we also cover \PT problem where $\hh_k$ consists all graphs with no simple path of length $k$. 

\begin{lemma}
\label{lem:treewidth-exact}
The problems of $k$-\textsc{Treewidth / Pathwidth / Treedepth Vertex
  Deletion} parameterized by $k$ and the solution size $p$ admit exact
algorithms with running times  $2^{O((k+p)^2k)}n$, $2^{O((k+p)\cdot(k +
  \log(k+p)))}n$, and $2^{O((k+p)k)}n$ respectively.
\end{lemma}
\begin{proof}[Proof sketch.]
As these algorithms are variants of well-known previous algorithms, we briefly give a sketch of the proof here and give more detailed explanations in Section~\ref{sec:twd-deletion}.
For any graph $H$ we have $tw(H) \le pw(H) \le td(H)$.
Consider a solution $X \subseteq V(G)$ -- it satisfies $|X| \le p$ and $tw(G\setminus X) \le k$.
After adding $X$ to each bag of the tree decomposition for $G\setminus X$ we obtain a decomposition for $G$ with width at most $k+p$.
We can thus use the linear-time constant approximation algorithm for treewidth~\cite{treewidth-apx} to find a tree decomposition of $G$ with width $O(k+p)$ in time $2^{O(k+p)}n$.

The problem of finding a tree (or path) decomposition of width $k$ parameterized by the width $t$ of the input tree decomposition has been studied by~\cite{BK91} who gave a $2^{O(tk + t\log t))}n$-time algorithm for the pathwidth case
and a $2^{O(t^2k)}n$-time algorithm for the treewidth case.
A $2^{O(tk)}n$-time algorithm for finding a treedepth decomposition of width $k$ given a tree decomposition of width $t$ has been obtained by~\cite{treedepth-dynamic}.
We slightly modify these procedures to handle vertex deletion
and use them over the precomputed tree decomposition of width $O(k+p)$.
A more detailed construction is presented in Section~\ref{sec:twd-deletion}. 
\end{proof}

\tvd*
\begin{proof}
For \textsc{Treewidth / Pathwidth / Treedepth Vertex Deletion}, we inject the bounds from Lemma~\ref{lem:treewidth-exact} into
Corollary~\ref{cor:main_vertex}.
To handle \textsc{Path Transversal} observe that
the $k$-path-free graphs have treedepth bounded by $k$ and therefore also treewidth bounded by $k$~\cite{treedepth-dynamic}.
There is an exact algorithm for \PT with running time $f_k(n,p) = O(k^p n)$, where $p$ is the bound on the solution size~\cite{Lee18}.
The claim follows again from Corollary~\ref{cor:main_vertex}.
\end{proof}

\subsection{Applications in Bidimensionality}
\label{sec:bidimensionality}
In this section we show how to obtain better guarantees
over planar graphs or, more generally, over graphs with excluded minor.
The main insight from bidimensionality we rely on is the following lemma allowing to truncate the solution candidate by increasing the working treewidth moderately.

\begin{lemma}[\cite{FLRS11}, Corollary 1]
\label{lem:minor_treewidth}
Let $G$ be a $M$-minor-free graph, $X \subseteq V(G)$, and $tw(G \setminus X) \le t$.
Then for any $\epsilon > 0$ there exists a set $X' \subseteq V(G)$ such that $|X'| \le \epsilon |X|$ and $tw(G \setminus X') = O_M({t}/{\epsilon})$, where the hidden constant depends on the excluded minor $M$.
Moreover, for given $G,X,\epsilon$, the set $X'$ can be constructed in polynomial time, however with a slightly worse guarantee $tw(G \setminus X') = O_M(\frac{t\log t}{\epsilon})$.
\end{lemma}
\begin{proof}
%The result in \cite{FLRS11} is stated without controlling dependency on $t$ and we retrace their reasoning for the reader's convenience. 
%\elnote{Slightly toned down}
We retrace the proofs in \cite{FLRS11} to give explicit dependence on $t$. 
%The result in \cite{FLRS11} is stated without explicit dependence on $t$ and we retrace their reasoning for the reader's convenience.
Their Corollary~2 says that if $G$ is $M$-minor-free and $tw(G \setminus X) \le t$, then $tw(G) = O_M(t\sqrt{|X|})$, i.e., $M$-minor-free graphs have "truly sublinear treewidth" with $\lambda = 1/2$.
For such a family and with assumptions as above, Lemma~1 guarantees that there exists $\gamma_t(\eps)$ and a set $X' \subseteq V(G),\, |X'| \le \epsilon |X|$,
such that every connected component $C$ of $G \setminus X'$ satisfies $|C \cap X| \le \gamma_t(\eps)$ and $|N(C)| \le \gamma_t(\eps)$.
Moreover, the proof indicates that $\gamma_t(\eps) = O_M\left(\left(\frac{t}{\eps}\right)^2\right)$ in the existential variant and $\gamma_t(\eps) = O_M\left(\left(\frac{t\log t}{\eps}\right)^2\right)$ in the constructive variant.
Injecting this bound into their Corollary~1 entails the claim.
\end{proof}
%\mwnote{Added a guide how to read their proof + added $\log t$ for the constructive version.}
%\elnote{Independently checked and the parameter is correct.} 
We now prove our meta-theorem for bidimensional problems. 
Roughly, a problem is bidimensional if the solution value for the problem on a $k \times k$ grid is $\Omega(k^2)$,
which is true for \HVD when $\hh$ is a class of graphs with bounded treewidth.
Introduced in Demaine et al.~\cite{DFHT05}, it has been a unifying theory for many algorithms in minor-free graphs. 
Demaine and Hajiaghayi~\cite{DH05} and later Fomin et al.~\cite{FLRS11} designed EPTASes for a large class of bidimensional problems.
For \HVD on $M$-minor free graphs $G$, there is an uniform algorithm runs in time $n^{f(\hh, M)}$~\cite{FLRS11}, and a non-uniform algorithm that runs in time $g(\hh, M) \cdot n^{O(1)}$~\cite{FLMS12}. 

As previously observed in~\cite{FLMS12}, the main bottleneck of the running time was reducing treewidth, so our algorithm for \TVD can be used to obtain improved running time for all bidimensional problem considered in~\cite{FLRS11}. 
We formally present EPTASes for \HVD with explicit running times that are uniform as long as the promised exact algorithm is. 
The only hidden factor we do not keep track of comes from the grid obstruction for excluded minor $M$.

%We formally prove our results for with explicit running time for \HVD, which captures many vertex minimization bidimensional problems studied in the literature. 
%We build another meta-theorem atop Theorem~\ref{thm:main_vertex},  providing an EPTAS for \HVD as long as we can solve the problem with a bound on the solution size.
%This framework captures most vertex minimization bidimensional problems studied in the literature~\cite{FLRS11}. 

%, which is moderate when we work on planar graphs.
%, has unknown dependency on $k$ in the running time.
%Here we additionally take advantage of the bidimensionality to show that the dependence on $k$ can be even subexponential.

\eptas*
\begin{proof}
Let us start with finding an $O(\log t)$-approximate solution $X$ with Corollary~\ref{cor:main_vertex} in time $f(n, O(t\log t)) + n^{O(1)}$.
Since $tw(G\setminus X) \le t$, we can use the constructive variant of Lemma~\ref{lem:minor_treewidth} with $\epsilon' = O(\epsilon / \log t)$ to find $X'$ such that $|X'| \le \epsilon' |X| \le \frac{\epsilon}{2}\cdot \OPT$ and $tw(G \setminus X') \le O_M(\frac{t\log^2 t}{\epsilon})$.
Though tree composition of $G \setminus X'$ is not explicitly given, we can find a decomposition of width $O(tw(G \setminus X'))$ in time $2^{O(tw(G \setminus X'))}\cdot n$~\cite{treewidth-apx} .

We apply the constructive variant of Lemma~\ref{lem:cut} to graph $G\setminus X'$ with $R = X$ and $\delta = O_M\left(\frac{t\log^3 t}{\epsilon^2}\right)$.
By choosing an appropriate constant, we obtain set $Y$
of size at most $\frac{O(tw(G\setminus X'))}{\delta}\cdot |X| \le \frac{\epsilon}{2}\cdot \OPT$ such that each connected component $C$ of $G \setminus (X' \cup Y)$ satisfies $|C \cap X| \le \delta$.
We launch the exact algorithm for \HVD on each component
with bound $\OPT \le \delta$ in total time $f(n,\delta)$ and return the sum of solutions together with $X' \cup Y$.
\end{proof}

\iffalse
\begin{proof}
Let $X$ be the optimal solution and $\delta$ be a parameter depending on $M$ to be fixed later.
Since $tw(G \setminus X) \le t$,
Lemma~\ref{lem:minor_treewidth} guarantees that there is a set
$X' \subseteq V(G)$ of size $\delta\cdot \OPT$ that forms a solution to \textsc{$(\frac{\rho_M\cdot t}{\delta})$-Treewidth Vertex Deletion}.
We take advantage of Corollary~\ref{cor:tvd} to find an $O(\log\frac{\rho_M\cdot t}{\delta})$ approximate solution $Y$ to this problem.
When we substitute $\delta = \epsilon / \log{\frac{t}{\epsilon}}$, the size of $Y$ can be estimated by
\begin{equation*}
\Big(\log\frac{t}{\epsilon} + \log\log\frac{t}{\epsilon} + \log\rho_M + O(1)\Big)\cdot \Big(\epsilon / \log{\frac{t}{\epsilon}}\Big)\cdot \OPT = O_M(\epsilon)\cdot \OPT.
\end{equation*}

By choosing sufficiently small $\delta = \Theta_M({\epsilon}/{\log \frac{t}{\epsilon}})$,
we can enforce $|Y| \le \epsilon \cdot \OPT$.
We have $tw(G \setminus Y) = O_M(\frac{t}{\epsilon} \log \frac{t}{\epsilon})$ so we can launch the exact algorithm for \HVD on $G \setminus Y$ and return it together with $Y$.
\end{proof}
\fi

We illustrate some applications of the above theorem.
For some problems, we additionally take advantage of the bidimensionality to show that the dependence on $k$ can be even subexponential.
%The only hidden factor we do not keep track of comes from the grid obstruction for excluded minor $M$, which is moderate when we work on planar graphs.

\begin{corollary}
\label{cor:minor_pt}
\PT and \kVS admit an EPTAS on $M$-minor-free graphs with running time $\exp\{O_M\big(\frac{\sqrt{k}\log^4 k}{\epsilon^2}\big)\}\cdot n\log n + n^{O(1)}$.
\end{corollary}
\begin{proof}
Graphs with excluded minor $M$ that are $k$-path free or have each component size bounded by $k$ have treewidth of order $O_M(\sqrt{k})$~\cite{bidimensionality}.
There is an exact algorithm for \PT with running time $f_k(n,p) = O(k^p n)$, where $p$ is the bound on the solution size~\cite{Lee18}.
The same approach works for \kVS{}: as long as the graph has a component of size at least $k+1$ we can find a connected subgraph of size $k+1$.
At least one of its vertices must belong to the solution so we can perform branching with $k+1$ direct recursive calls and depth at most $p$.
The claim follows from Theorem~\ref{thm:eptas} with $f\left(n, O_M\left(\frac{\sqrt{k}\log^3 k}{\epsilon^2}\right)\right)= \exp\{O_M\big(\frac{\sqrt{k}\log^4 k}{\epsilon^2}\big)\}\cdot n$.
\end{proof}

\begin{corollary}
\label{cor:minor_tvd}
The \PVD and \TDVD problems admit EPTASes on $M$-minor-free graphs with running
time $\exp\{O_M\big(\frac{k^2\log^3 k}{\epsilon^2}\big)\} \cdot n\log n
+ n^{O(1)}$. Also,
\TVD admits an analogous result with running time $\exp\{O_M\big(\frac{k^3\log^6 k}{\epsilon^4}\big)\} \cdot n\log n + n^{O(1)}$.
\end{corollary}
\begin{proof}
We apply Lemma~\ref{lem:treewidth-exact}, providing an exact algorithm for these problems, to
Theorem~\ref{thm:eptas}.
The respective running times of these routines are:
\begin{OneLiners}
\item $f^{pw}_k(n,p) =2^{O\left((k+p)\cdot(k + \log(k+p))\right)} \cdot n$, 
\item $f^{td}_k(n,p) =2^{O\left((k+p)k)\right)} \cdot n$, and
\item $f^{tw}_k(n,p) =2^{O\left((k+p)^2k)\right)} \cdot n$. 
\end{OneLiners}
This completes the proof.
\end{proof}

%%% Local Variables:
%%% mode: latex
%%% TeX-master: "main"
%%% End:

%% file: kses.tex
For \kES and \kSES the best previous approximation algorithm gave $O(\log k)$-approximation~\cite{Lee18}. 
We present an improved $(2 + \eps)$-approximation algorithm for \kES in Section~\ref{sec:kes}, 
and give two extensions to \kSES in Section~\ref{sec:kses_local} and~\ref{sec:kses_uml} with (almost) the same approximation ratio. 
In Section~\ref{sec:fed}, we apply these algorithms for \HED and study further applications in Section~\ref{subsec:edge_applications}.
Finally, in Section~\ref{sec:inapprox}, we prove inapproximability results for \kES for $k = 3$ which implies inapproximability for all edge deletion problems considered in this paper. 

\subsection{\kES}
\label{sec:kes}
We first give a $(2 + \eps)$-approximation algorithm for \kES that runs in time $2^{O(k \log k)} n^{O(1)}$, proving the first part of Theorem~\ref{thm:kes}.
%\mwnote{fixed some minor bugs here} 
It can be proved in two ways. 
After showing that we can assume that the maximum degree is $O(k)$ without loss of generality, 
the first proof is a reduction to \UML studied by Kleinberg and Tardos~\cite{KT02},
where 2-approximation is achieved via the standard LP relaxation. 
The second proof is based on a direct local search algorithm,
which creates a new component with at most $k$ vertices whenever it improves the overall cost. 

These two proofs lead to two approximation algorithms for \kSES with similar approximation ratios that run in time 
$f(k, \deg(G)) \cdot n^{O(1)}$ and $n^{f(k)}$ respectively.
%\mwnote{this is a bit inaccurate because one of them is 2-apx}
The first algorithm is based on the reduction to \UML extended by the technique of {\em important separators}~\cite{MR14}, 
and the second algorithm extends the local search in the second proof by efficiently computing the best local move.
%\elnote{Some citations}

For \kES, we present the local search based algorithm.
We start by noting that considering only bounded degree graphs suffices,
since in the optimal solution, large degree vertices will lose almost all incident edges.

\begin{claim}
\label{claim:deg_reduction}
For any $\eps > 0$, 
an $\alpha$-approximation algorithm for \kES that runs in time $f(k, \deg(G)) n^{O(1)}$
implies an $(\alpha + \eps)$-approximation algorithm for \kES that runs in time $f(k, 2 k / \eps) n^{O(1)}$. 
\end{claim}
\begin{proof}
Suppose that there exists an $\alpha$-approximation algorithm for \kES that runs in time $f(k, \deg(G)) n^{O(1)}$.
To solve \kES for a unbounded degree graph, given a graph $G = (V(G), E(G))$,
we remove all the edges incident on vertices whose degree is more than $2 k / \eps$. 

Let $S^* \subseteq E(G)$ be the optimal solution. 
Every vertex $v$ in $G \setminus S^*$ has degree at most $k - 1$, 
so for $v$ with $\deg(v) > 2 k / \eps$, 
the above operation deletes at most $k - 1$ edges not in $S^*$, 
which is at most $\eps / 2$ fraction of edges in $S^*$ incident to $v$. 
Therefore, this operation deletes at most $\eps |E(S^*)| = \eps \cdot \OPT$ edges overall outside $S^*$. 
Running the bounded degree algorithm on the resulting graph proves the claim. 
\end{proof}

Now we give an algorithm for \kES.
\begin{lemma} [(i) of Theorem~\ref{thm:kes}]
There is a $(2 + \eps)$-approximation algorithm for \kES that runs in time $2^{O(k \log (k/\eps))}n^{O(1)}$. 
\end{lemma}

\begin{proof}
%Given $\eps > 0$, apply Claim~\ref{claim:deg_reduction} to assume that the degree is bounded by $2k / \eps$ by deleting at most $\epsilon \cdot \OPT$ extra edges. 
Our local search algorithm maintains the partition $(C_1, \dots, C_m)$ of $V(G)$ where $|C_i| \leq k$ for each $i$. 
This corresponds to deleting edges in $E(C_1, \dots, C_m)$.
In each iteration, the algorithm considers every possible part $C \subseteq V$ of size at most $k$ such that the induced subgraph $G[C]$ is connected. 
There are at most $n\cdot \deg(G)^{k}$ such components to consider. 
For each $C$, we consider the new partition where $C$ is added to the partition, and each previous part $C_i$ becomes $C_i \leftarrow C_i \setminus C$.
(Delete empty part from the partition.) 
If the new partition cuts fewer edges, implement this change and repeat until there is no possible improvement. 
Since each iteration strictly improves the current solution, the total running time is bounded by $\deg(G)^k n^{O(1)}$. 

Let $(C_1, \dots, C_m)$ be the resulting partition output by the local search and $(C^*_1, \dots, C^*_{m^*})$ be the optimal partition. 
For each $C^*_i$, either $C^*_i$ is a part in $(C_1, \dots, C_m)$, or the local move with $C^*_i$ does not improve $(C_1, \dots, C_m)$.
In the latter case, 
as the local improvement with $C^*_i$ newly deletes edges in $\partial C_i^* \setminus E(C_1, \dots, C_m)$ and 
restores currently deleted edges in $E(G[C^*_i]) \cap E(C_1, \dots, C_m)$, we can conclude that 
\[
|\partial C_i^*  \setminus E(C_1, \dots, C_m)| \geq |E(G[C^*_i]) \cap E(C_1, \dots, C_m)|. 
\]
Note that in the former case, the above is trivially satisfied. 
If we add the above for every $i = 1, \dots, m^*$, 
the left-hand side is two times the number of edges in $E(C^*_1, \dots, C^*_{m^*})  \setminus E(C_1, \dots, C_m)$,
and the right-hand side is the number of edges in $E(C_1, \dots, C_m) \setminus E(C^*_1, \dots, C^*_{m^*})$. 
Therefore, 
\[
2|E(C^*_1, \dots, C^*_{m^*})  \setminus E(C_1, \dots, C_m)| \geq 
|E(C_1, \dots, C_m) \setminus E(C^*_1, \dots, C^*_{m^*})| \Rightarrow 2\OPT \geq |E(C_1, \dots, C_m)|.
\]
Therefore, this algorithm runs in time $2^{O(k \log \deg(G))} n^{O(1)}$ and gives a $2$-approximation. 
Applying Claim~\ref{claim:deg_reduction}, 
for any $\eps > 0$, we have a $(2 + \eps)$-approximation algorithm that runs in time $2^{O(k \log (k/\eps))} n^{O(1)}$ for general graphs.
\end{proof}

\subsection{\kSES in Time $n^{k+O(1)}$}
\label{sec:kses_local}
Note that the above local search algorithm, without the degree reduction step, also implies a $2$-approximation algorithm in time $n^{k+O(1)}$,
since there are at most $n^k$ subsets of $V$ of size at most $k$. 
For \kSES where each part can contain much more than $k$ vertices as long as it has at most $k$ vertices from $R$, even the degree bound 
does not yield a polynomial bound on the number of choices we need to consider in the local search algorithm. 
For example, given a subset $R' \subseteq R$ with $|R'| \leq k$, there can be exponentially many $C \subseteq V$ such that $C \cap R = R'$ and 
$G[C]$ is connected. 

The modified local search algorithm for \kSES, in each iteration, finds the {\em best} local improvement over all possible subsets. 
The following lemma shows that it can be done in polynomial time. It immediately proves (ii) of Theorem~\ref{thm:kes} which gives a $2$-approximation algorithm for \kSES in time $n^{k+O(1)}$. 

\begin{lemma}
Let $(C_1, \dots, C_m)$ be a partition of $V$ and $\emptyset \neq R' \subsetneq R$. 
There is a polynomial time algorithm to find $C^* \subseteq V$ that minimizes the cost
$|E(C_1 \setminus C, \dots, C_m \setminus C, C)|$ over every set $C$ that satisfies $R \cap C = R'$. 
\label{lem:kses_mincut}
\end{lemma}
\begin{proof}
From $G$, merge all vertices in $R'$ to a vertex $s$, and merge all vertices in $R \setminus R'$ to a vertex $t$, while creating parallel edges if needed. Let $G_1$ be the resulting graph. 
Let $B \subseteq E(G_1)$ be the edges cut by the current solution. Call them {\em blue} edges. 
Finding the best $C^*$ in $G$ is equivalent to finding the best $s$-$t$ cut $(S, V(G_1) \setminus S)$ in $G_1$ $(s \in S)$
that minimizes 
\[
| (\partial_{G_1} S) \setminus B | + | B \setminus G_1[S] |,
\]
which is exactly the total cost of the new partition. The first term is the number of edges that are newly deleted by adding $S$ to the partition, 
and the second term is the number of the previously deleted edges minus the number of the undeleted edges in $S$. 

We find the minimum $S$ by reducing to the classic \STCUT problem. Starting from $G_1$, we do the following operations to obtain $G_2$.
\begin{itemize}
\item For each non-blue edge, do not change anything.
\item For each blue edge $e = (u, v)$ with $u, v \in V \setminus \{s, t \}$ we introduce a new vertex $t_e$ and replace $(u, v)$ by three edges $(s, t_e), (u, t_e), (v, t_e)$. 
\begin{itemize}
\item If $u, v \in S$, we can put $t_e$ to $S$ and do not cut any edge. 
\item If $u \notin S$ and $v \notin S$, we cut one edge by putting $t_e$ to $V \setminus S$. 
\item If $|S \cap \{ u, v \}| = 1$, we cut one edge by putting $t_e$ to $S$. 
\end{itemize}
\item For each blue edge $e = (s, v)$ with $v \neq t$, do not change anything.
\begin{itemize}
\item If $v \in S$, we do not cut any edge. 
\item If $v \notin S$, we  cut one edge. 
\end{itemize}
\item For each blue edge $e = (u, t)$ with $u \neq s$, keep this edge and add one more edge $(s, u)$. 
\begin{itemize}
\item We cut one edge whether $u \in S$ or not. 
\end{itemize}
\end{itemize}
Note that for each $(u, v) \in E(G_1)$, we cut exactly one edge in $G_2$ if (1) it is non-blue and cut by $S$, or (2) it is blue and not completely contained in $S$. 
This is exactly the objective function that we want to minimize in $S_1$. 
Therefore, the minimum $s$-$t$ cut in $G_2$ gives the optimal solution in $G_1$, which in turn gives the optimal local improvement $C^*$ in $G$. 
\end{proof}

\subsection{\kSES Parameterized by Degree}
\label{sec:kses_uml}

In this section, we provide a $(2+\eps)$-approximation algorithm for \kSES parameterized by $k$ and the maximum degree of the graph, proving (iii) of Theorem~\ref{thm:kes}. 
Throughout this section, we will fix $\eps>0$ and the maximum degree $d$ of the graph. Our algorithm has three main steps. First, we will reduce our search space of solutions to \kSES to a smaller set of \textit{canonical} solutions, which behave more nicely. In particular, in each canonical solution $S\subseteq E(G)$, every connected component of $G \setminus S$ containing a vertex in $R$ has a small number of edges leaving the component. We will show that there always exists a canonical solution of size $\le(1+\eps)\OPT$. Then, we will find a $2$-approximation to the best canonical solution by reducing to the \UML problem, which we will define later. Formulating the \UML instance requires another ingredient, the concept of \textit{important cuts}, a tool popular in FPT algorithm design.

We begin with canonical solutions.

\begin{defn}
A solution $S\subseteq E(G)$ to \kSES is called $\eps$-\textit{canonical} if, for each connected component $C\subseteq V(G)$ of $G \setminus S$ satisfying $C\cap R\ne\emptyset$, we have $|\partial C|\le 2k\deg(G)/\eps$.
\end{defn}

\begin{obsr}
There exists a $\eps$-canonical solution with size at most $(1+\eps)\OPT$.
\end{obsr}

\begin{proof}
Consider the optimal solution $S^*$, which we modify as follows. For each component $C\subseteq V(G)$ with $|\partial C|>2k\deg(G)/\eps$, further delete all edges incident to each vertex in $R\cap C$. We delete $\le k\deg(G)$ edges, which can be charged evenly to the boundary edges of $C$, so that each edge gets charged $\le \eps/2$. Every edge gets charged twice, so the total number of additional edges deleted is $\le\eps\cdot \OPT$. It is clear that $S^*$ with these additional edges deleted is $\eps$-canonical.
\end{proof}

At this point, we are looking for a solution that separates the graph into pieces with a small number of vertices in $R$ and small boundary. Our next step is to provide a ``cover'' for all possible such pieces, which we will use in our \UML reduction. In particular, we look for a set $\mathcal C\subseteq 2^V$ of subsets of vertices of small size such that every piece $C\subseteq V$ that we might possibly look for satisfies $C\subseteq C'$ for some $C'\in\mathcal C$.

\begin{lemma}\label{lem:SmallCover}
There exists a set $\mathcal C\subseteq 2^V$ of subsets of vertices of size $4^{O(kM)}\cdot n$ such that (1) every subset $C\in \mathcal C$ satisfies $1\le|C\cap R|\le k$, and (2) for every connected component $C\subseteq V$ satisfying $1\le|C\cap R|\le k$ and $|\partial C|\le M$, we have $C\subseteq C'$ for some $C'\in\mathcal C$.% with $|\partial C' | \leq |\partial C|$.
\end{lemma}

Note that we will set $M:=2k\deg(G)/\eps$ when applying Lemma~\ref{lem:SmallCover} later on. The main ingredient in the proof of Lemma~\ref{lem:SmallCover} is the concept of \textit{important cuts} in a graph, along with a result bounding the number of important cuts of a bounded size.

\begin{defn}[Important cut]
For vertices $s,t\in V(G)$, an $s$--$t$ cut is a subset $X\subseteq V(G)$ of vertices such that $s\in X$ and $t\notin X$. An \textit{important $s$--$t$ cut} is an $s$--$t$ cut $X\subseteq V(G)$ with the following two additional properties:
\begin{enumerate}
\item The induced graph $G[X]$ is connected.
\item There is no $s$--$t$ cut $X'\subseteq V(G)$ such that $|\partial X'|\le|\partial X|$ and $X\subseteq X'$.
\end{enumerate}
\end{defn}

\begin{theorem}[\cite{cygan2015parameterized}, Theorem 8.11]\label{thm:ImportantCuts}
For fixed vertices $s,t\in V(G)$ and integer $p\ge0$, there are at most $4^p$ important $s$--$t$ cuts of size at most $p$. Moreover, all of these can be enumerated in time $O(4^p\cdot n^{O(1)})$.
\end{theorem}

Using this theorem, we now prove Lemma~\ref{lem:SmallCover}.

\begin{proof}[Proof (Lemma~\ref{lem:SmallCover})]
Fix a vertex $s\in R$. Our goal is to establish a set $\mathcal C_s\subseteq 2^V$ of size $4^{O(kM)}$ such that (1) every subset $C\in\mathcal C_s$ satisfies $1\le|C\cap R|\le k$, and (2) for every connected component $C\subseteq V$ satisfying $s\in C$, $1\le|C\cap R|\le k$, and $|\partial C|\le M$, we have $C\subseteq C'$ for some $C'\in\mathcal C_s$. Then, we can take $\mathcal C:=\bigcup_{s\in R}\mathcal C_s$ of size $4^{O(kM)}\cdot n$, which satisfies the lemma.

Consider the following construction. Take the graph $G$, add a new vertex $t$, and for each vertex $v\in R\setminus \{s\}$, add $M+1$ parallel edges connecting $v$ and $t$; call the new graph $H$. Apply Theorem~\ref{thm:ImportantCuts} on $H,s,t$ with $p:=(k-1)(M+1)+M$, giving $4^{O(kM)}$ important $s$--$t$ cuts. Let $\mathcal C_s$ be these cuts; we now show that this set works. For every connected component $C\subseteq V$, we have \[|\partial_H C|=|\partial_GC|+|C\cap(R\setminus\{s\})|\cdot (M+1).\] If $C$ contains $s$ and satisfies $|C\cap R|\le k$ and $|\partial_G C|\le M$, then $|\partial_HC|\le M+(k-1)(M+1)\le p$. Therefore, either $C\in\mathcal C_s$, or there is an important cut $X\in\mathcal C_s$ such that $|\partial_HX|\le|\partial_HC|$ and $C\subseteq X$. In the latter case, $X$ cannot contain $\ge k$ vertices in $R\setminus\{s\}$, since that would mean $|\partial_HX|\ge|X\cap(R\setminus\{s\})|\cdot(M+1)\ge k(M+1)>p\ge |\partial_HC|$, contradicting the assumption that  $|\partial_HX|\le|\partial_HC|$. Therefore, $X$ contains $\le k$ vertices in $R$, including $s$. Finally, we have $|X\cap (R\setminus\{s\})|\ge|C\cap(R\setminus\{s\})|$, so \[|\partial_GX|=|\partial_H X|-|X\cap(R\setminus\{s\})|\cdot (M+1)\le|\partial_H C|-|C\cap(R\setminus\{s\})|\cdot (M+1)=|\partial_GC|.\] Hence, the subset $X\in\mathcal C_s$ satisfies the conditions of the lemma for $C$.
\end{proof}

We invoke Lemma~\ref{lem:SmallCover} with $M:=2k\deg(G)/\eps$ and compute the corresponding set $\mathcal C$. The last step in the algorithm is to reduce the problem to an instance of \UML.

\begin{defn}[\UML]
Given a graph $G=(V,E)$, a set of labels $L$, and cost matrix $A\in \mathbb R_+^{V\times L}$ where entry $A_{v,\ell}$ is the cost of labeling vertex $v\in V(G)$ with label $\ell\in L$, the \UML problem is to label each vertex in $V(G)$ with exactly one label in $L$ that minimizes
\[ \sum_{v\in V} A_{v,l(v)}+\sum_{(u,v)\in E} \mathbf 1_{l(u)\ne l(v)}, \]
where $l(v)$ is the label of vertex $v$ and $\mathbf 1_{l(u)\ne l(v)}$ equals $1$ if the labels of $u$ and $v$ are different, and $0$ otherwise.
\end{defn}

%While the general metric labeling problem is more difficult, the uniform case admits a 2-approximation algorithm from~\cite{kleinberg2002approximation}.

\begin{theorem}[\cite{KT02}]
There is a polynomial-time $2$-approximation algorithm for \UML.
\end{theorem}

We reduce to \UML as follows: the labels are the subsets in $\mathcal C$ along with a dummy label, called $\bot$. For each vertex $v\in V(G)$ and label $C\in\mathcal C$, the cost $A_{v,C}$ is $0$ if $v\in C$, and $\infty$ otherwise. That is, we do not allow a vertex to be labeled by a subset that does contain that vertex. For label $\bot$, the cost is $A_{v,\bot}=0$ if $v\notin R$, and $\infty$ otherwise. That is, we do not allow a vertex in $R$ to be labeled $\bot$. Observe that this \UML instance has size $4^{O(kM)}n^{O(1)}=2^{O(k^2\deg(G)/\eps)}n^{O(1)}$.

It is clear that any solution to this \UML instance is a valid solution to \kSES: the components containing vertices in $R$ are precisely the maximal connected components of the same label, and each such component must have $\le k$ vertices in $R$. Moreover, the best $\eps$-canonical solution $S$ for \kSES can be transformed into a solution for \UML with the same solution value as follows: for each connected component $C\subseteq V$ in $G \setminus S$ with a vertex in $R$, take a set $C'\in\mathcal C$ with $C\subseteq C'$ and color all vertices in $C$ with label $C'$; for connected components without a vertex in $R$, label all their vertices $\bot$. Thus, we can compute a $2$-approximation to \UML and obtain a solution within factor $2$ of the best $\eps$-canonical solution, or within factor $2(1+\eps)$ of the optimum. Of course, we can make the approximation factor $2+\eps$ by resetting $\eps\gets\eps/2$.

%, with the size of the \UML instance being parameterized by both $k$ and the maximum degree $\deg(G)$ of the graph, as follows.

%First, we will reduce our search space of solutions to a smaller set of canonical solutions which behave nicely. We show that there always exists a canonical solution of size $\le(1+\eps)\OPT)$.

%Then, we reduce finding a canonical solution to a \UML problem instance. Here, we will use , specifically the concept of important cuts.

%% file: hed.tex
\label{sec:fed}

Our main theorem for \HED is the following. Unlike the vertex deletion, our algorithm uses $\deg(G)$ as an extra parameter, and it is an interesting open problem whether we can remove this dependence. 

\mainedge*
\begin{proof}
As for the vertex deletion version, our algorithm maintains a feasible solution and iteratively tries to improve it. 
Let $S^* \subseteq E(G)$ be an optimal solution to \HED and let $R_E \subseteq E(G)$ be the solution.
From $G$ and $R_E$, construct a graph $G'$ where we {\em subdivide} each edge $e = (u, v) \in R_E$; 
formally, create a new vertex $r_e$ and replace $(u, v)$ by $(u, r_e)$ and $(v, r_e)$. 
Let $R'_V := \{ r_e : e \in R_E \} \subseteq V(G')$. 
From $S^*$, let $S' \subseteq E(G')$ be such that 
for each edge $e = (u, v) \in S^*$, 
\begin{itemize}
\item If $e \notin R_E$, put $e$ to $S'$.
\item If $e \in R_E$, arbitrarily choose one endpoint (say $u$) and put $(u, r_e)$ to $S'$. 
\end{itemize}
By construction, $|S'| = |S^*|$ and $|R'_V| = |R_E|$. 

Note that $G' \sm S'$ can be obtained from $G \sm S^*$ by subdividing edges (when $e \in R_E \sm S^*$) and add degree one vertices (when $e = (u, v) \in R_E \cap S^*$, $G' \sm S'$ additionally has $(v, r_e)$ compared to $G \sm S^*$). 
Both operations do not increase the treewidth, so the fact that graph $G \sm S^*$ has its treewidth bounded by $t$ implies that 
$G' \sm S'$ has its treewidth bounded by $t$. 
By Lemma~\ref{lem:cut} with $\delta = t\cdot\beta$ ($\beta$ to be chosen later) guarantees that there exists 
a set $X'_V \subseteq V(G'),\, |X'_V| \le \frac{ |R'_V|}{\beta}$ so that each connected component in $(G' \sm S') \sm X'_V$
has at most $t\cdot\beta$ vertices from $R'_V$.
Since subdividing edges does not increase the maximum degree, $\deg(G') \leq \deg(G)$. 
Let $X'_E \subseteq E(G')$ be the set of edges incident on $X'_V$. Then $|X'_E| \le \frac{\deg(G) \cdot |R'_V|}{\beta} = \frac{\deg(G) \cdot |R_E|}{\beta}$ so that each connected component in $G' \sm (S' \cup X'_E)$ has at most $t\cdot\beta$ vertices from $R'_V$.

We launch the $(2 + \eps)$ approximation for \kSES with $k = t\cdot\beta$ on $G'$ and $R'_V$.
It returns a set $Y' \subseteq E(G')$ of size at most

$$(2 + \eps) \cdot |S' \cup X'_E| \le (2 + \eps)(\OPT + \frac{\deg(G)|R_E|}{\beta})$$
with each connected component of $G' \sm Y'$ having at most $t\cdot \beta$ vertices from $R'_V$.
Let $Y$ be $Y'$ {\em projected} back to $G$; formally, $Y := \{ e = (u, v) \in E(G) : e \in Y' \mbox{ or } (u, r_e) \in Y' \mbox{ or } (v, r_e) \in Y' \}$. 
Since $G' \sm Y'$ has at most $t \cdot \beta$ vertices from $R'_V$, 
$G \sm Y$ has at most $t \cdot \beta$ edges from $R_E$. 

Since $\hh$ is hereditary, $R_E \cup Y$ is a valid solution and we have a bound $t\cdot\beta$ on the solution size for each connected component.
We thus can solve \HED on each component $C \subseteq G \sm Y$ in time $f(n, t \cdot \beta)$.
We know that $C \cap S^*$ is a feasible solution for each $C$, so the sum of solution sizes is bounded by $|S^*| = \OPT$.
Therefore, we obtain a feasible solution of size 
$$(2 + \eps)\cdot \bigg(\OPT + \frac{\deg(G)|R_E|}{\beta}\bigg) + \OPT.$$

Let $\beta = \deg(G) / \eps$, so that the above quantity becomes
$$(3 + \eps)\OPT + \eps (2 + \eps) |R_E|.$$
This becomes better than $(1 - \eps)|R_E|$ when 
\[
|R_E| (1 - \eps (3 + \eps)) \geq (3 + \eps) \OPT \Leftrightarrow
|R_E| \geq \frac{(3 + \eps)}{(1 - \eps (3 + \eps))} \OPT = (3 + O(\eps)) \OPT.
\]

Therefore, if we repeat this iteration until there is no improvement by a factor of $(1 - \eps)$, the final solution is guaranteed to be within $(3 + O(\eps)) \OPT$. 
The running time is 
\[
\min(2^{O(k^2\deg(G)/\eps)} n^{O(1)}, n^{k+O(1)}) = 
\min(2^{O(k^2\deg(G)/\eps)} n^{O(1)}, n^{O(t \deg(G) / \eps)})
\]
for $\kSES$ with $k = t\beta = t\deg(G) / \eps$ plus
$f(t \deg(G) / \eps) n^{O(1)}$ for the final step. There can be at most $\log(n / \eps)$ iterations. 
\end{proof}

\subsection{Applications for Bounded Degree Graphs}
\label{subsec:edge_applications}
We prove the corollaries of Theorem~\ref{thm:main_edge} introduced in Section~\ref{sec:intro}.

\fed*

\begin{proof}
For \FED, we use the Polynomial Grid Minor theorem~\cite{CC16} which says that if $G$ does not have a planar graph $F$ as a minor, the treewidth of $G$ is bounded by $|V(F)|^{O(1)}$.
\FED admits a linear-time exact algorithm parameterized by the solution size~\cite{Courcelle90}  so the assumptions of Theorem~\ref{thm:main_edge} are satisfied.%\elnote{Did Bodlaender solve FED too?}

Due to the result of Robertson and Seymour~\cite{graph-minors-xx}, every minor-closed class can be represented
as $\ff$-minor-free graphs for some finite family $\ff$.
If the treewidth in $\hh$ is additionally bounded, then
at least one of graphs in $\ff$ must be planar.
Therefore \HED reduces to \FED.
\end{proof}

We also present implications of Theorem~\ref{thm:main_edge} to the \NPkSd problem studied by Bansal et al.~\cite{BRU17}. For a fixed $k = O(1)$, an instance of \NPkSd is an instance of $\phi$ of \kSAT where the {\em factor graph} $H$ of $\phi$ is almost planar.
 
Formally, given a $k$-CNF formula $\phi$ with $n$ variables and $m$ clauses, 
the factor graph of $\phi$ is a bipartite graph $H = (A, B)$ where $A$ contains a vertex for every variable appearing in $\phi$, $B$ contains a vertex for every clause appearing in $\phi$, and a clause-vertex $C$ is connected to a variable-vertex $x$ if and only if $x$ belongs to $C$. 
As an instance of \NPkSd, the factor graph of $\phi$ is promised to be a planar graph with $\delta m$ additional edges for some $\delta > 0$. 

Bansal et al.~\cite{BRU17} proved that for any $\eps > 0$, there is an algorithm that achieves $(1 + O(\eps + \delta \log m \log \log m))$-approximation in time $m^{O(\log \log m)^2 / \eps}$. We prove that if the degree of the factor graph is bounded, we can obtain an improved algorithm.
Note that \kSAT with the maximum degree $O(1)$ has been actively studied and proved to be APX-hard (e.g., \textsc{$3$-SAT(5)}) for general factor graphs. 

\bru*

\begin{proof}
Recall that we treat $k$ as an absolute constant. 
Given an instance $\phi$ of \NPkSd, the factor graph $H$ of $\phi$ has $\Theta(m)$ vertices. 
Deleting $O(\delta m)$ edges from $H$ will make $H$ planar, and additionally deleting $O(\eps m)$ edges will make its treewidth bounded by $O(1/\eps)$. 

We apply Corollary~\ref{cor:fed} to delete $O((\eps + \delta)m)$ edges of $H$ to reduce its treewidth to $O(1 / \eps)$. 
Its running time is $f(\eps, \deg(H)) \cdot m^{O(1)}$. Delete all clauses that lost at least one of the incident edges. 
We deleted $O((\eps + \delta)m)$ clauses. Now that the treewidth is bounded by $O(1/\eps)$, apply an exact algorithm for \kSAT that runs in time $2^{O(1/\eps)} \cdot m^{O(1)}$~\cite{KM96}. 
\end{proof}
%\elnote{Slight improvement for Ind. Set? As Anupam mentioned, even theirs can be improved using degree-reduction, so do not bother?}

%% file: inapprox.tex
\label{sec:inapprox}

\newcommand{\cT}{\mathcal{T}}
\newcommand{\cF}{\mathcal{F}}
\newcommand{\cI}{\mathcal{I}}
\newcommand{\cS}{\mathcal{S}}

We end this section by proving the inapproximability of \kES, as stated below.

\hardnessedge*

In fact, we will prove an NP-hardness of approximation for a slightly different partitioning problem, which is sometimes referred to as \emph{Partitioning into Triangles (PIT)}. In PIT, we are given a graph $G = (V, E)$ and the goal is to find the largest collection of disjoint triangles, i.e., disjoint subsets of vertices $S_1, \dots, S_k \subseteq V$ such that each $S_i$ is of size three and induces a 3-clique. We will show the following hardness of approximation for PIT:

\begin{lemma}
\label{lem:pit}
There exists $\varepsilon > 0$ such that it is NP-hard, given a graph $G = (V, E)$ with $\deg(G) = 4$, to distinguish between the following two cases, where $n$ denotes $|V|$:
\begin{itemize}
\item (Completeness) The vertex set $V$ can be partitioned into $n/3$ disjoint triangles.
\item (Soundness) Every collection of disjoint triangles has size less than $(1 - \varepsilon)n/3$.
\end{itemize}
\end{lemma}

PIT is a classic NP-complete problem (see~\cite{GJ79}), and it should be remarked that Kann~\cite{Kann91} already showed that the problem is Max SNP-hard when $\deg(G) = 6$; indeed, this already suffices for proving Theorem~\ref{thm:kes_inapprox} if we relax the degree requirement to $\deg(G) = 6$. Nevertheless, even if we want $\deg(G) = 4$, the proof is still simple, and the reduction is in fact exactly the same as that of van Rooij et al.~\cite{RNB13}, who showed the NP-hardness of (the exact version of) PIT when $\deg(G) = 4$. The authors of~\cite{RNB13} also showed that PIT becomes polynomial time solvable when $\deg(G) \leq 3$ and, hence, the degree requirement cannot be improved in Lemma~\ref{lem:pit}. Nevertheless, we are not aware of either an efficient algorithm or hardness result for the $\deg(G) = 3$ case for \kES, and we leave that as an open question.

Before we prove Lemma~\ref{lem:pit}, let first us state how it implies Theorem~\ref{thm:kes_inapprox}.

\begin{proof}[Proof of Theorem~\ref{thm:kes_inapprox}]
The reduction is trivial: we keep the input $G$ to PIT as it is, and set $k = 3$. Moreover, let $c = 1 + \varepsilon / 4$ where $\varepsilon$ is the constant from Lemma~\ref{lem:pit}.

(Completeness) Suppose that there exists a partition of $V$ into $n/3$ triangles $S_1, \dots, S_{n/3}$. There are $n$ uncut edges with respect to this partition, and hence it cuts exactly $|E| - n$ edges.

(Soundness) Suppose that every collection of disjoint triangles has size less than $(1 - \varepsilon)n/3$. Consider any partition of $V$ into disjoint subsets $T_1, \dots, T_k$, each of size at most three. Our assumption implies that less than $(1 - \varepsilon)n$ vertices are adjacent to two uncut edges. Hence, the total number of uncut edges is less than $(1 - \varepsilon)n + \varepsilon n/2 = (1 - \varepsilon/2)n$, and the number of cut edges is more than $|E| - (1 - \varepsilon/2)n$.

The ratio between the two cases is more than $$\frac{|E| - n}{|E| - (1 - \varepsilon/2)n} = 1 + \frac{\varepsilon n/2}{|E| - (1 - \varepsilon/2)n} \geq 1 + \frac{\varepsilon n/2}{2n} = c,$$ where the inequality comes from $\deg(G) \leq 4$. This concludes our proof.
\end{proof}

We now turn our attention back to the proof of Lemma~\ref{lem:pit}. As stated earlier, we exactly follow the reduction of van Rooij et al.~\cite{RNB13}. They reduce from \emph{Max 1-in-3SAT} problem, in which we are given a 3CNF formula and the goal is to find an assignment that assigns exactly one literal to be true in each clause. Since we want to prove hardness of approximation, we will need hardness of approximation of \emph{Max 1-in-3SAT}, which is well-known\footnote{The result stated in Lemma~\ref{lem:inapprox-exact3sat} is folklore, although we are not aware of it being stated in this form before. However, it is quite easy to see that it is true, as follows. First, recall that Max-3SAT is NP-hard to approximate even on bounded degree instances~\cite{Has00,Tre01}. Then, observe that we can use the reduction of Schaefer~\cite{Sch78} from 3SAT to 1-in-3SAT, which is approximation-preserving and also preserves boundedness of the degrees.} and can be stated as follows.

\begin{lemma}
\label{lem:inapprox-exact3sat}
There exists $\delta > 0$ such that it is NP-hard, given a 3CNF formula\footnote{For the purpose of our proof, each clause in a 3CNF formula contains exactly three literals.} such that each variable appears in at most $d = O(1)$ clauses, to distinguish between the following two cases:
\begin{itemize}
\item (Completeness) There is an assignment such that exactly one literal in each clause is true.
\item (Soundness) Every assignment satisfies less than $(1 - \delta)$ fraction of clauses.
\end{itemize}
\end{lemma}

We will also need the gadgets from~\cite{RNB13}, which can be summarized as follows. Since this is exactly the same as those used in~\cite{RNB13}, we do not provide full constructions of them here; we refer the readers to Lemma 8 of~\cite{RNB13} for more details. Illustrations of the gadgets can be founded in Figure~\ref{fig:gadgets}, which is reconstructed (with slight modifications) from Figure 5 of~\cite{RNB13}.

\begin{figure}
\begin{subfigure}[c]{.40\textwidth}
\centering
\begin{tikzpicture}[scale=0.8]
\node[draw,thick] (I1) at (-1, 0) [circle] {$I_1$};
\node[draw,thick] (I2) at (1, 0) [circle] {$I_2$};
\node[draw,thick] (O1) at (0, 2) [circle] {$O_1$};
\node[draw,thick] (O2) at (-1.5, -2) [circle] {$O_2$};
\node[draw,thick] (O3) at (1.5, -2) [circle] {$O_3$};
\draw[thick] (I1) -- (I2);
\draw[thick] (I1) -- (O1);
\draw[thick] (I1) -- (O2);
\draw[thick] (I1) -- (O3);
\draw[thick] (I2) -- (O1);
\draw[thick] (I2) -- (O2);
\draw[thick] (I2) -- (O3);
\end{tikzpicture}
\subcaption{}
\label{subfig:fan}
\end{subfigure}
\begin{subfigure}[c]{.40\textwidth}
\centering
\begin{tikzpicture}[scale=0.8]
\node[draw,thick] (I1) at (-1, 0) [circle] {$I$};
\node[draw,thick] (I2) at (1, 0) [circle] {$I$};
\node[draw,thick] (F) at (0, 2) [circle] {$F$};
\node[draw,thick] (T1) at (-3, 0) [circle] {$T$};
\node[draw,thick] (T2) at (3, 0) [circle] {$T$};
\node[draw,thick] (T3) at (-2, -2) [circle] {$T$};
\node[draw,thick] (T4) at (2, -2) [circle] {$T$};
\draw[thick,dashed] (I1) -- (I2);
\draw[thick,dashed] (I1) -- (F);
\draw[thick,dashed] (I2) -- (F);
\draw[thick] (I1) -- (T1);
\draw[thick] (I1) -- (T3);
\draw[thick] (T1) -- (T3);
\draw[thick] (I2) -- (T2);
\draw[thick] (I2) -- (T4);
\draw[thick] (T2) -- (T4);
\end{tikzpicture}
\subcaption{}
\label{subfig:cloud}
\end{subfigure}
\caption{Illustration of a fan and a cloud. A fan is depicted in Figure~\ref{subfig:fan}. A (4, 1)-cloud is depicted in Figure~\ref{subfig:cloud}; here the true vertices are marked by ``T'', the false vertex by ``F'' and the inner vertices by ``I''. The dashed triangle corresponds to the true collection (as defined in Definition~\ref{def:cloud}), whereas the remaining two triangles correspond to the false collection.}
\label{fig:gadgets}
\end{figure}
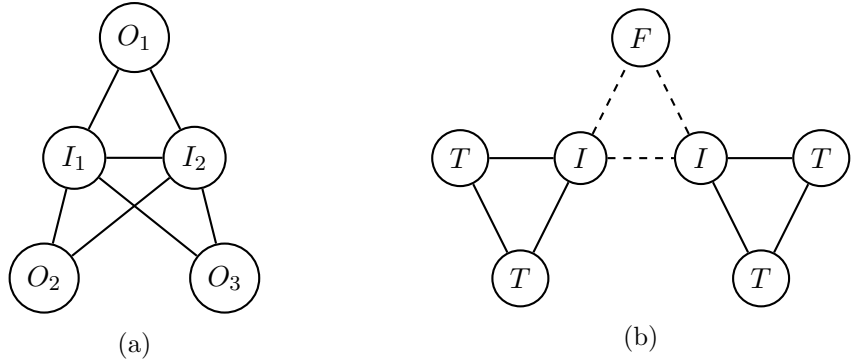

\begin{defn}
A \emph{fan} is a graph of five vertices $O_1, O_2, O_3, I_1, I_2$ and seven edges: $\{I_1, I_2\}$ and $\{I_i, O_j\}$ for all $i \in [2]$ and $j \in [3]$. In other words, it is a union of three triangles having one edge $\{I_1, I_2\}$ in common. We call $O_1, O_2, O_3$ \emph{outer vertices} and $I_1, I_2$ \emph{inner vertices} of the fan. 
\end{defn}

\begin{defn} \label{def:cloud}
For $a, b \in \mathbb{N}$, an \emph{$(a, b)$-cloud} is a graph of $2(a + b) - 3$ vertices that satisfies the following properties:
\begin{itemize}
\item The vertices can be divided into three groups: $a$ \emph{true vertices}, $b$ \emph{false vertices} and $(a + b) - 3$ \emph{inner vertices}.
\item Each true vertex and each false vertex has degree two.
\item There are only two different collections of disjoint triangles that contain all inner vertices. In one collection, every false vertex is included but none of the true vertices are included; we call this collection the \emph{true} collection. In the other collection, every true vertex is included but none of the false vertices are included; we call this collection the \emph{false} collection
\end{itemize}
\end{defn}

\begin{lemma}[\cite{RNB13}]
For every $a, b \in \mathbb{N}$ such that $a \equiv b \mod 3$, an $(a, b)$-cloud exists.
\end{lemma}

We are now ready to prove Lemma~\ref{lem:pit}.

\begin{proof}[Proof of Lemma~\ref{lem:pit}]
As stated earlier, we follow the reduction of van Rooij et al.~\cite{RNB13} from Max 1-in-3SAT to PIT, although we will have to be slightly more careful in the analysis, as we want to not only prove hardness for the exact version but also the approximate version of the problem.

Van Rooij et al.'s reduction can be described as follows:
\begin{itemize}
\item First, notice that we can assume without loss of generality that the number of occurrences of each literal is divisible by three; this can be easily ensure by duplicating all the clauses twice.
\item For each variable $x$, let $a(x)$ be the number of occurrences of the literal $x$ and $b(x)$ be the number of occurrences of the literal $\neg x$. We create an $(a(x), b(x))$-cloud for each variable $x$.
\item We create a fan for each clause $C$. For each literal in the clause, we identify one outer vertex of the fan to a vertex corresponding to that literal in the cloud of the variable. Note that, for each variable $x$, since there are $a(x)$ and $b(x)$ vertices corresponding to $x$ and $\neg x$ respectively, the identification can be done in such a way that each literal vertex is identified with exactly one vertex from a clause cloud, which also ensures that the graph has maximum degree four.
\end{itemize}
Finally, let $\varepsilon = \delta/(8d + 8)$. Moreover, let us use $N$ and $M$ to denote the number of variables and the number of clauses of the 3CNF formula respectively, and $n$ to denote the number of vertices of the resulting graph. Notice that, from the bounded degree assumption, we have $M \leq d N / 3$. Moreover, from the sizes of each gadgets, we have $n \leq 2M + \sum_{x} 2(a(x) + b(x)) = 8M$.

(Completeness) Suppose that there exists an assignment $\phi$ such that each clause contains exactly one true literal. Then, we can define our partition as follows. For each variable $x$, we pick the true or false collection for the $x$-cloud based on the value $\phi(x)$. For each clause, we pick the triangle with two inner vertices and the outer vertex corresponding to the true literal. It is clear that this is indeed a partition of vertices into disjoint triangle as desired.

(Soundness) We will show the contrapositive; suppose that there exists $k \geq (1 - \varepsilon)n/3$ disjoint triangles $S_1, \dots, S_k$. We call a variable $x$ \emph{good} if the triangles restricted to only those entirely contained in the $x$-cloud is either the true collection or the false collection. Notice that, if each inner vertex of $x$-cloud is in at least one of the selected triangles, then $x$-cloud must be good, since the inner vertices of the $x$-cloud are not adjacent to any vertices outside of the cloud. However, there are at most $\varepsilon n$ vertices outside of the disjoint triangles, meaning that at most $\varepsilon n \leq 8\varepsilon M$ variables are not good.

Next, we call a clause $C$ \emph{good} if (1) the three variables whose literals are in $C$ are good and (2) the inner vertices of $C$-fan are in at least one of the selected triangles. Notice that there are at most $d(8\varepsilon M) = 8 \varepsilon d M$ clauses that violate (1). More, again, since there are at most $\varepsilon n \leq 8\varepsilon M$ vertices outside of the union of the triangles, at most $8\varepsilon M$ clauses violate (2). Hence, all but at most $8 \varepsilon (d + 1) M \leq \delta M$ clauses are good.

We will define an assignment $\phi$ as follows. For each good $x$, we define $\phi(x)$ to be true if the triangles correspond to the true collection of the $x$-cloud, and we let $\phi(x)$ be false otherwise. For the remaining $x$'s, we assign $\phi(x)$ arbitrarily. It is easy to see that $\phi$ satisfies all the good clauses; this is simply because exactly one literal in each good clause $C$, the one whose triangle with the two inner vertices in the $C$-fan is selected, is set to true. Hence, $\phi$ satisfies all but $\delta M$ clauses, which concludes our proof.
\end{proof}

%% file: twd-deletion.tex
\section{Details of parameterized width-reduction algorithms}
\label{sec:twd-deletion}
In this section we give a more detailed proof sketch of Lemma~\ref{lem:treewidth-exact}.
Indeed, we show a more general claim that the algorithms from~\cite{BK91} and~\cite{treedepth-dynamic} can be extended to finding a
tree-decomposition (or path / treedepth-decomposition) of small width $k$ given a decomposition of a larger width $t$ can be extended to handle vertex deletion ($t = O(k + p)$ in Lemma~\ref{lem:treewidth-exact}).
The reasoning presented here is a proof sketch and we focus on explaining why the arguments from previous works remain valid in the  extended versions of the algorithms.
Since the main claim of~\cite{BK91} is just a linear running time for fixed $k$ and $t$, we need to explicitly bound the number of states in the dynamic programming routines.

\subsection{Original routines for finding small tree or path decompositions}

\begin{defn}
For an integer sequence $a$, its typical sequence $\tau(a)$ is obtained by iterating the following operations until none is possible anymore:
\begin{OneLiners}
\item[1.] removal of repetitions of consecutive elements,
\item[2.] removal of a subsequence $a_{i+1}, \ldots a_{j-1}$ satisfying
$\forall_{i<k<j} \, a_i \le a_k \le a_j$ or \mbox{$\forall_{i<k<j} \, a_j \le a_k \le a_i$}. 
\end{OneLiners}
\end{defn}

The sequence $\tau(a)$ is uniquely defined~\cite[Lemma~3.2]{BK91}.

\begin{lemma}[Lemma 3.3 and 3.5 in \cite{BK91}]
There are $O(4^k)$ typical sequences of integers in $[0,k]$.
The length of each one is at most $2k+1$.
\end{lemma}

\begin{defn}
For integer sequences $a,b$ we write $a \prec b$ if one can extend them to sequences $a',b'$ of equal length by adding consecutive repetitions, such that $a' \le b'$ on each index.
\end{defn}

The relation $\prec$ is transitive, and $a \prec b$ holds if and only if $\tau(a) \prec \tau(b)$~\cite[Lemma~3.7 and~3.10]{BK91}.

All algorithms in question work on a given tree decomposition of width at most $t$.
We can assume that this is a \emph{nice tree decomposition}, i.e., a binary tree $T$, in which every node $x$ is assigned a bag $B_x \subseteq V(G)$ and belongs to one of the following types: 
\begin{OneLiners}
\item[1.] \emph{start:} the node is a leaf of $T$ with only one vertex in $B_x$,
\item[2.] \emph{join:} the node has two children $y,z$ satisfying $B_x = B_y = B_z$,
\item[3.] \emph{introduce:} the node has one child $y$ and $B_x$ is formed by adding one vertex to $B_y$,
\item[4.] \emph{forget:} the node has one child $y$ and $B_x$ is formed by removing one vertex from $B_y$.
\end{OneLiners}

We define $T_x$ be the subtree of $T$ rooted at $x$, and $G_x$ to be a subgraph of $G$ induced by vertices introduced in $T_x$.
A partial tree (path) decomposition is a decomposition of a  subgraph $H$ of $G_x$ of width at most $k$.
%(however in this subsection we only work with case $H = G_x$).

\begin{defn}
For a partial path decomposition $Y = (Y_1, Y_2, \dots, Y_m)$ we define its \emph{restriction} with respect to the set $B_x$
to be $Z = (Y_1 \cap B_x, Y_2 \cap B_x, \dots, Y_m \cap B_x)$.
Let $1=t_1 < t_2 < \dots < t_q$ be all indices for which $Z_{t_i - 1} \ne Z_{t_i}$ and let sequence $a^i$ indicate $(|Y_{t_i}|, |Y_{t_i+1}|, \dots |Y_{t_{i+1}-1}|)$.
The \emph{characteristic} of $Y$
is given by the sequence $(Z_{t_i})_{i=1}^q$ called the \emph{interval model}, and the list of sequences
$(\tau(a^i))_{i=1}^q$.
\end{defn}

%Now we are ready to describe states in the dynamic programming routine over $T$.
%For every partial path decomposition of an induced subgraph $H$ of $G_x$ 

The number of possible interval models in a node  is $2^{O(t\log t)}$ and the maximal length of such a model is $2t+3$ (Lemma 3.1 in~\cite{BK91}).
Therefore the number of all possible characteristics in a node is bounded by $2^{O(t\log t)} \cdot O(4^{kt}) = 2^{O(kt + t\log t)}$ (Lemma 4.1 in~\cite{BK91}).

\begin{defn}
For a partial tree decomposition $Y$ we define its \emph{restriction} $Z$
with respect to the set $B_x$ again by intersecting each bag with $B_x$.
A leaf of a restriction is called \emph{maximal} if its bag is not contained in a bag of any other node.
The \emph{trunk} of such a decomposition is obtained by iteratively removing leaves that are not maximal and then replacing all nodes of degree 2 with edges.
Each edge $e$ in a trunk induces a partial path decomposition of some subgraph of $G_x$---let $Z^e$ denote its interval model and $a^e$ the associated list of typical sequences.
%\elnote{each edge $e$ induces a partial path decomposition of length $2$?}
The \emph{characteristic} of $Y$
is given by the trunk of $Z$, a family of interval models for each edge in the trunk, and a family of lists of typical sequences for each edge in the trunk.
\end{defn}

Since the number of leaves in a trunk is at most $t$, the number of its nodes is $O(t)$ and the number of such trees is $2^{O(t\log t)}$.
Since we need to store as many as $t$ typical sequences for each edge of the trunk, the number of all possible characteristics in a node $x$ is $2^{O(t^2k)}$.

For two characteristics of partial path decomposition we say the one majorizes another, written as $((Z_i), (a^i)) \prec ((Z'_i), (b^i))$, if $Z_i = Z'_i$ for all $i$ and $a^i \prec b^i$ for all $i$.
The same notion applies to partial tree decomposition when the trunks are the same and majorization occurs on each edge of the trunk.

In the original algorithm one maintains a \emph{full set} of characteristics for each node $x$ describing the minimal interface between $G_x$ and the rest of the graph.
We can only store characteristics that admit a respective partial tree (path) decomposition.
Moreover if there is a partial decomposition of $G_x$ with characteristic $\mathcal{C}$,
then the full set must contain a characteristic that is being majorized by~$\mathcal{C}$.
This ensures that if there is a partial decomposition that can be extended to a full decomposition, then the interface contains its characteristic or a characteristic of another extendable partial decomposition.

\subsection{Extension to vertex deletion}

Consider the given nice tree decomposition $T$ as defined in the previous subsection.
We can assume (by adding extra \emph{forget} nodes) that the bag in the root is empty.
We build a directed acyclic graph $T'$ with nodes given by triples $(x,X,\ell)$, where $x$ is a node from $T$, $X$ is a subset of $B_x$, and integer $\ell$ indicates how many vertices we have deleted in $G_x$.

Since $T$ is nice, each graph vertex $v \in V$ has one tree node that {\em forgets} it and possibly many nodes that {\em introduce} it. 
Note that the forget node for $v$ is an ancestor of all its introduction nodes, and $v$ only appears in the subtree of $T$ rooted at its forget node. 
Formally, $(x, X, \ell)$ stores the dynamic programming state (i.e., a set of characteristics) to compute a tree decomposition of width $k$ as in the previous subsection, 
where the bag containing $x$ becomes $X$ instead of $B_x$ while at most $\ell$ vertices are deleted among vertices in $G_x$ whose forget node is in the subtree rooted at $x$.
Deleted vertices will be counted at its forget node. 
%\elnote{I added the last sentence for intuition. Modify accurately if you want. }

For a \emph{start} node $x$ we just add $(x,B_x,0)$ to $T'$.
If $x$ is a \emph{join} node with children $y,z$ we add an edge from $(x,X,\ell)$ to $(y,X,\ell)$ and $(z,X,\ell)$ for all $X \subseteq B_x,\, 0 \le \ell \le n$.
For a node $x$ that introduces a vertex $v$ into its child $y$'s bag, we create  nodes $(x,X,\ell)$ connected to $(y,X \setminus \{v\} ,\ell)$ for all $v \in X \subseteq B_x,\, 0 \le \ell \le n$, as well as \emph{dummy} nodes $(x,X \setminus \{v\} ,\ell)$ connected to $(y,X \setminus \{v\} ,\ell)$.
Dummy nodes represent an introduction operation that has been canceled.
Finally for a \emph{forget} node that removes a vertex $v$ from its child
$y$, we have an edge from $(x,X,\ell)$ to $(y,X \cup
\{v\},\ell)$ for all $X \subseteq B_x,\, 0 \le \ell \le n$, and to
$(y,X ,\ell - 1)$ for all $X \subseteq B_x,\, 1 \le \ell \le n$.
%\elnote{Is $\ell + 1$ actually $\ell - 1$? I thought $\ell$ increases as we go up the tree (and ignore $v$ as here) } 
This represents branching into scenarios where $v$ will or will not be deleted.

The number of nodes in $T'$ is bounded by $V(T)\cdot 2^t \cdot n$. 
The following claim can be easily verified from the construction.
\begin{claim}
Fix a node $x \in T$ and consider a directed subtree $T'_x$ of our DAG such that:
\begin{OneLiners}
\item[1.] for each $y \in T_x$, there is exactly one $y' := (y, B'_y, \ell'_y) \in T'_x$,
\item[2.] for each $y$ and its child $z$ in $T_x$, there is an edge from $y'$ to $z'$ in $T'_x$. 
\end{OneLiners}
Then $T'_x$ is a nice tree decomposition of $G_x \setminus S$, where $S$ is the union of $B_x \setminus B'_x$ and some subset of size $\leq \ell$ in $G_x$. 
For each vertex $v \in V(G)$, $v$ is in $T'$ if and only if all $y$ with $v \in B_y$ satisfies $v \in B_{y'}$. 
\end{claim}
%\elnote{Anyway we know that we will delete $\leq t$ vertices ($t = O(k + p)$), so $n$ can be replaced by $t$?}
%Observe that fixing one of the outgoing edges in each \emph{forget} node induces a tree decomposition of a subgraph of $G$.
Given this observation, we can fill out the dynamic programming tables for the DAG as we did for $T$. 
We can run the original routine on $T'$ that handles \emph{introduce} and \emph{join} nodes as before, only with set $B_x$ replaced by $X$.
In a dummy node we just copy results from the child.
For a \emph{forget} node $(x,X ,\ell)$, we compute the characteristics coming from branch $(y ,X \cup \{v\},\ell)$ as in the original routine and then add all characteristics copied from node $(y,X ,\ell - 1)$.

\begin{theorem}
\TVD and \PVD parameterized by $k$ and the width $t$ of the given tree decomposition admit exact algorithm with running times respectively
$2^{O(t^2k)}n$ and $2^{O(tk+t\log t)}n$.
\end{theorem}
\begin{proof}[Proof sketch.]
Sections 4.3, 4.4, and 4.5 in~\cite{BK91} describe how to compute the full set for partial path decompositions in nodes of type $\emph{join, forget, introduce}$ in time polynomial with respect to the size of characteristics' space.
Sections 5.3, 5.4, and 5.5 deal with the same cases for partial tree decompositions.

We can run these routines on $T'$ and compute the sum of characteristics' sets when branching in \emph{forget} nodes.
The invariant of the full set for node $(x,X,\ell)$ becomes: each characteristic is induced a partial decomposition of subgraph $H$ of $G_x$, such that $V(H) \cap B_x = X$ and $|V(H)| + \ell = |V(G_x)|$ and for each such partial decomposition the full set contains one being majorized by it.

The size of $T'$ is quadratic in $n$, so an explicit implementation of the algorithm would be burdened with a quadratic time.
%\elnote{If my above comment about replacing $n$ by $t$ is right, can be removed? }
However we can observe that for each characteristic for fixed $x,X$ we only need to remember the smallest $\ell$ for which it is feasible.
Thus, we can work with only $V(T)\cdot 2^t$ nodes and store the smallest feasible $\ell$ for each characteristic in a full set.
\end{proof}

\begin{theorem}
\TDVD parameterized by $k$ and the width $t$ of the given tree decomposition admits an exact algorithm with running time
$2^{O(tk)}n$ .
\end{theorem}
\begin{proof}[Proof sketch.]
We apply the same reasoning as for treewidth deletion.
The number of states necessary to remember in a single node is bounded explicitly by $2^{O(tk)}$ (Lemma 15 in~\cite{treedepth-dynamic}) and the running time is analyzed in Lemma 17.
\end{proof}